\documentclass{article}

    \PassOptionsToPackage{numbers, sort}{natbib}
\usepackage[preprint]{neurips_2026}

\usepackage[utf8]{inputenc} 
\usepackage[T1]{fontenc}    
\usepackage{hyperref}       
\usepackage{url}            
\usepackage{booktabs}       
\usepackage{amsfonts}       
\usepackage{nicefrac}       
\usepackage{microtype}      

\usepackage{amsthm}        
\usepackage{amsmath}
\usepackage{graphicx} 
\usepackage{caption}  
\usepackage{float}    
\usepackage{subcaption} 
\usepackage{booktabs} 
\usepackage{enumitem}

\usepackage{algorithm}      
\usepackage{algpseudocode}  
\usepackage{algorithmicx}   

\usepackage{rotating}
\usepackage{tabularx}
\usepackage{booktabs}
\usepackage{adjustbox}
\usepackage{multirow}
\usepackage{dsfont}

\usepackage[T1]{fontenc}
\usepackage[utf8]{inputenc}
\usepackage{amsmath,amssymb}
\usepackage{algorithm}
\usepackage{algpseudocode}
\usepackage{xspace}

\newcommand{\ours}{\textsc{DecentMem}\xspace}

\newcommand{\tabref}[1]{Table~\ref{#1}}
\newcommand{\eqnref}[1]{\text{Eq.}~(\ref{#1})}
\newcommand{\appref}[1]{Appendix~\ref{#1}}
\newcommand{\figref}[1]{Fig.~\ref{#1}}
\newcommand{\secref}[1]{\S\ref{#1}}

\usepackage{booktabs}   
\usepackage{multirow}   
\usepackage{graphicx}   
\usepackage[table,xcdraw]{xcolor} 
\usepackage{amsmath}    

\newtheorem{theorem}{Theorem}

\newtheorem{corollary}{Corollary}
\newtheorem{Assumption}{Assumption}
\newtheorem{remark}{Remark}

\definecolor{upcolor}{HTML}{F28B71}    
\definecolor{downcolor}{HTML}{48C1B6}  
\definecolor{bestbg}{HTML}{DEE6F0}     
\definecolor{tagbg}{HTML}{B4B4B4}      

\newcommand{\best}[1]{\cellcolor{bestbg}#1}

\usepackage{caption}
\usepackage{wrapfig}
\usepackage[most]{tcolorbox}

\hypersetup{
    colorlinks=true,
    linkcolor=red,
    citecolor=cyan,
    filecolor=magenta,      
    urlcolor=magenta,
}

\definecolor{caseTitleBlue}{RGB}{177,221,240}
\definecolor{caseSectionBlue}{RGB}{218,232,252}
\definecolor{caseBorder}{RGB}{120,145,160}
\definecolor{caseGreen}{RGB}{0,190,0}

\definecolor{gmemTitleBlue}{RGB}{190,238,240}
\definecolor{gmemSectionBlue}{RGB}{218,232,252}
\definecolor{gmemBorder}{RGB}{110,130,145}
\definecolor{gmemRed}{RGB}{255,0,0}

\newtcolorbox{gmemorycasebox}[1]{
    enhanced,
    breakable,
    colback=white,
    colframe=gmemBorder,
    colbacktitle=gmemTitle,
    coltitle=black,
    fonttitle=\bfseries\large,
    title={#1},
    boxrule=0.8pt,
    arc=2pt,
    left=6pt,
    right=6pt,
    top=6pt,
    bottom=6pt,
    toptitle=6pt,
    bottomtitle=6pt,
    lefttitle=8pt,
    righttitle=8pt,
    before skip=10pt,
    after skip=12pt,
    borderline west={2pt}{0pt}{gmemBorder}
}

\definecolor{decentTitle}{HTML}{DDF4EC}
\definecolor{decentBar}{HTML}{EEF9F4}
\definecolor{decentBorder}{HTML}{3A8F73}
\definecolor{decentLogBg}{HTML}{F6FCF9}
\definecolor{decentLogText}{HTML}{1E7A5B}

\definecolor{gmemTitle}{HTML}{FCE8E8}
\definecolor{gmemBar}{HTML}{FFF3F3}
\definecolor{gmemBorder}{HTML}{B85C5C}
\definecolor{gmemLogBg}{HTML}{FFF8F8}
\definecolor{gmemLogText}{HTML}{A33A3A}

\definecolor{caseText}{HTML}{222222}
\definecolor{caseMuted}{HTML}{666666}

\newenvironment{gmemcontent}{%
    \par
    \begingroup
    \small
    \setlength{\parindent}{0pt}%
    \setlength{\parskip}{3pt}%
    \leftskip=5pt
    \rightskip=5pt
    \vspace{4pt}%
}{%
    \par
    \vspace{4pt}%
    \endgroup
}

\usepackage{fvextra}

\usepackage{lineno}
\usepackage{titlesec}
\titlespacing\section{0pt}{0pt}{0pt}
\titlespacing\subsection{0pt}{0pt}{0pt}
\titlespacing\subsubsection{0pt}{0pt}{0pt}
\titlespacing\paragraph{0pt}{0pt}{4pt}

\title{Self-Evolving Multi-Agent Systems via\\Decentralized Memory}

\author{
Guangya Hao\textsuperscript{1}
\qquad
Yunbo Long\textsuperscript{1}
\qquad
Zhuokai Zhao\textsuperscript{2}\\[0.3em]
\textsuperscript{1}University of Cambridge
\quad
\textsuperscript{2}University of Chicago
}

\begin{document}

\maketitle

\begingroup
\renewcommand{\thefootnote}{}
\footnotetext{
\begin{tabular}{@{}l@{}}
Correspondence: \texttt{gh540@cam.ac.uk} and \texttt{zhuokai@uchicago.edu}
\end{tabular}
}
\endgroup

\begin{abstract}
Self-evolving multi-agent systems (MAS) have emerged as a promising route to LLM agents that continually improve from experience, with persistent memory at their foundation.
However, existing designs almost exclusively adopt a \textit{centralized} repository shared across agents, incurring communication and coordination overhead, raising privacy concerns, and collapsing agent diversity.
We propose \textbf{\ours}, a \textit{decentralized} memory framework in which each agent maintains its own dual-pool memory --- an exploitation pool of consolidated past trajectories and an exploration pool of LLM-generated candidates for unseen contexts.
The two pools are reweighted online based on stage-wise feedback  from an LLM-as-a-judge.
Theoretically, we prove that this design guarantees global reachability of the solution space and achieves $O(\log T)$ cumulative regret, matching the stochastic bandit lower bound up to constants.
In practice, across three MAS frameworks (AutoGen, DyLAN, AgentNet), three Qwen3 backbones (4B/8B/14B), two Gemma4 backbones (E2B/E4B)  and five benchmarks spanning math, code, QA, and embodied tasks, \ours improves average accuracy by up to 23.8\% over the strongest centralized memory baseline and by up to 52.5\% over the no-memory baseline, while reducing token usage by up to 49\%.

\end{abstract}

\section{Introduction}\label{sec:introduction}
\begin{wrapfigure}{r}{0.51\textwidth}
\vspace{-0.3in}
    \centering
    \includegraphics[width=\linewidth]{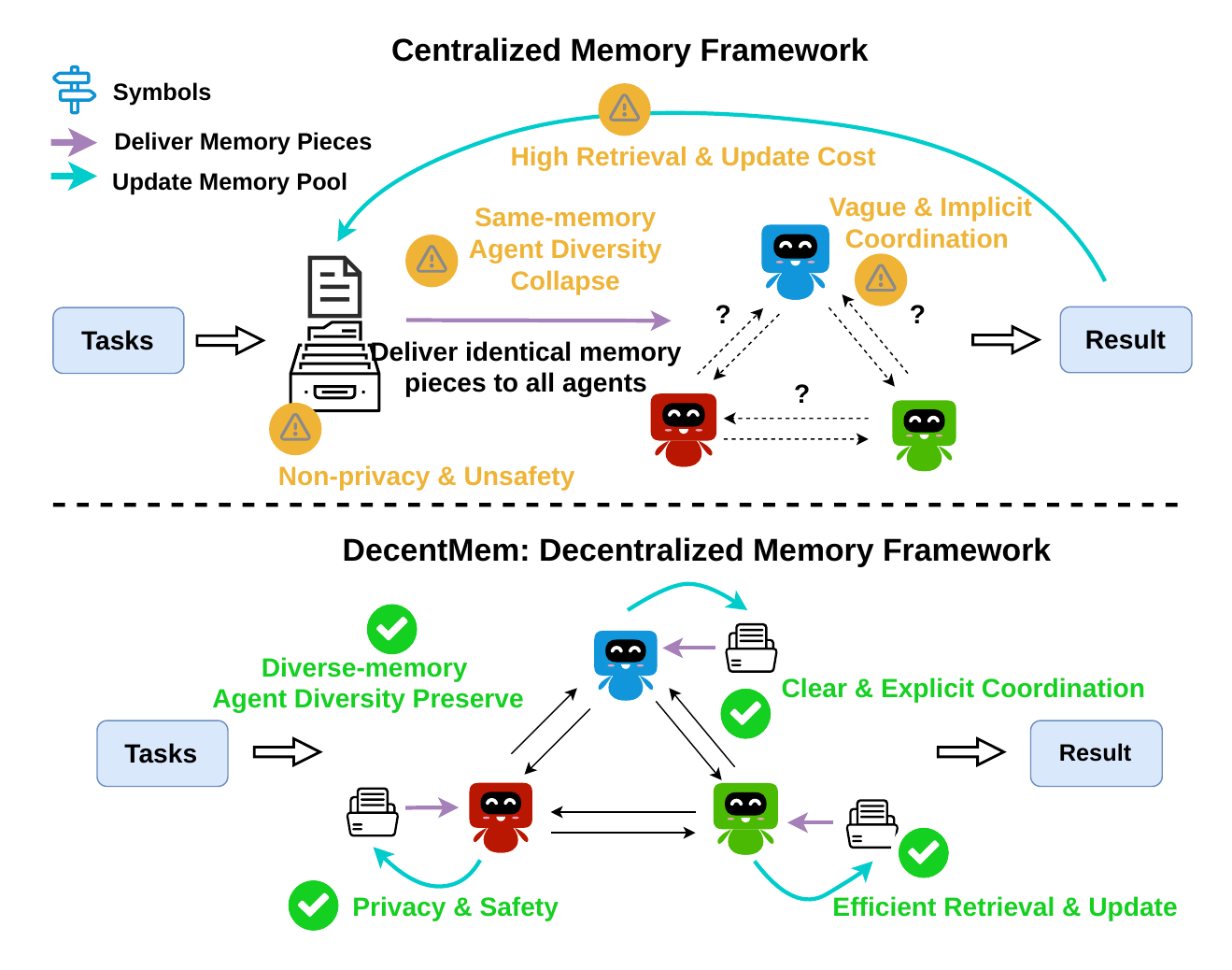}
    \caption{
        Compared to traditional centralized memory frameworks, \ours provides the first decentralized memory framework with adaptability to various MAS frameworks.
    }
    \label{fig:teaser}
    \vspace{-0.25in}
\end{wrapfigure}
Memory enables large language model (LLM) based multi-agent systems (MAS) to accumulate experience across tasks and improve over time.
The dominant design pattern in current MAS is a \textit{centralized} memory --- a single shared repository that every agent reads from and writes to --- inherited directly from single-agent memory architectures~\citep{zhong2024memorybank, packer2023memgpt, Chhikara2025Mem0BP}.
For example, MetaGPT's global message pool~\citep{Hong2023MetaGPTMP} and G-Memory's hierarchical shared graph~\citep{zhang2025g} both default to sharing everything across agents.
We argue, however, that centralization is the wrong default for multi-agent memory, and that \textit{decentralized} agent-private memory is both better motivated and empirically superior, as illustrated in \figref{fig:teaser}.

A key limitation of centralized MAS memory is that it is incompatible with self-evolution~\citep{Shinn2023ReflexionLA, Wang2023VoyagerAO, Chen2025ScalingAL}.
In multi-agent settings, self-evolution should preserve the role-complementary specialization that makes collaboration valuable in the first place, rather than drive the system toward a single shared strategy.
Centralized memory directly undermines this objective: when every agent retrieves from the same pool, their behaviors homogenize over time, and the system collapses toward a single dominant strategy regardless of how distinct the agents' initial roles were~\citep{Yu2026MultiAgentMF}.
The diversity that motivates instantiating multiple agents is gradually lost.
Centralization also imposes well-documented system-level costs, including high communication cost~\citep{hadfield2025multiagentresearch, hong2025context}, synchronization and coordination tax~\citep{cemri2025multi, Yu2026MultiAgentMF}, and privacy risks~\citep{Yang2025AgentNetDE}.

Motivated by these limitations, we propose \textbf{\ours}, a decentralized self-evolving memory framework for LLM-based MAS.
In \ours, each agent maintains its own dual-pool memory, which consists of an \textit{exploitation pool (E-pool)} of consolidated trajectories from past tasks, and an \textit{exploration pool (X-pool)} of LLM-generated candidates for unseen contexts.
The two pools serve complementary functions, where the E-pool drives a similarity-based local walk that reuses prior successful strategies, while the X-pool injects fresh probability mass into regions of the strategy space the agent has never visited, breaking the local-optimum trap that traps pure exploitation.
An LLM-as-a-judge evaluates each stage of the solution trajectory and re-weights the two pools online, allowing each agent to learn its own exploitation--exploration balance from feedback rather than committing to a fixed schedule.
Decentralization preserves agent diversity because the consolidated experience never leaves the agent that earned it, and the dual-pool structure preserves the capacity to escape that experience when it is not useful in the current task or scenario.

We evaluate \ours both theoretically and empirically.
Theoretically, we model multi-agent self-evolving search as a graph-structured random walk with heuristic teleportation and cast online routing between the two memory pools as a stochastic bandit problem. Under this formulation, we prove that \ours achieves global reachability over each agent's local solution subspace and attains $\mathcal{O}(\log T)$ cumulative regret, order-optimal against the $\Omega(\log T)$ bandit lower bound~\citep{Auer2002FinitetimeAO}.
Empirically, across three MAS frameworks (AutoGen~\citep{wu2024autogen}, DyLAN~\citep{liu2023dynamic}, AgentNet~\citep{Yang2025AgentNetDE}), five LLM backbones spanning dense and mixture-of-experts (MoE) architectures (Qwen3-4B/8B/14B~\citep{Yang2025Qwen3TR}, Gemma4-E2B/E4B~\citep{farabet2026gemma4}), and five benchmarks including mathematical reasoning, code generation, question answering, and embodied decision-making, we show that \ours improves average accuracy by up to 23.8\% over the strongest centralized memory baseline (G-Memory~\citep{zhang2025g}) and by up to 52.5\% over the no-memory baseline, while reducing token usage by up to 49\%.
Notably, the performance improvement widens as collaboration becomes more stochastic --- moving from AutoGen's pre-designed workflows~\citep{wu2024autogen} to AgentNet's opportunistic coordination~\citep{Yang2025AgentNetDE} --- indicating that decentralized memory preserves the divergent trajectories that centralized memory would otherwise collapse.

Our main contributions are summarized as follows:
\begin{itemize}[topsep=0pt,leftmargin=*,noitemsep]
    \item \textit{A decentralized dual-pool memory with online routing.} 
    We introduce \ours, the first decentralized memory framework for LLM-based MAS, in which each agent maintains a private exploitation pool of consolidated trajectories and a private exploration pool of LLM-generated candidates. 
    An online router, supervised by an LLM-as-a-judge, re-weights the two pools from stage-wise feedback so that each agent learns its own exploitation--exploration balance.

    \item \textit{Theoretical guarantees.} 
    We prove that \ours achieves global reachability over each agent's local solution subspace, ensuring no agent is permanently trapped in a locally suboptimal region. 
    We further prove that online pool routing attains $\mathcal{O}(\log T)$ cumulative regret, matching the $\Omega(\log T)$ bandit lower bound up to constants and converging to the optimal exploitation--exploration balance at an order-optimal rate.

    \item \textit{Comprehensive empirical validation.} 
    Across three MAS frameworks, five LLM backbones, and five benchmarks, we demonstrate that \ours outperforms the strongest centralized baselines by 8.6\% on average and uses up to 49\% fewer tokens, with the largest gains in the most stochastic collaboration regimes.

    %
\end{itemize}
\section{Related Work}\label{sec:related_work}

\paragraph{Single-agent memory.}
%
Memory is a central mechanism that enables LLM agents to accumulate experience and improve through repeated interaction with their environments.
Early systems mainly extended the effective context window by retrieving past interactions, as in MemoryBank~\citep{zhong2024memorybank} and MemGPT~\citep{packer2023memgpt}.
More recent work has moved beyond simple retrieval and introduced higher-level memory abstractions, including reflective memory formation in Generative Agents~\citep{park2023generative}, scalable personalized memory in Mem0~\citep{Chhikara2025Mem0BP}, dynamic memory organization in A-Mem~\citep{xu2025mem}, lightweight multi-stage memory management in LightMem~\citep{fang2025lightmem}, and efficient lifelong memory compression in SimpleMem~\citep{Liu2026SimpleMemEL}.
Despite these advances, existing memory designs are still largely developed for single-agent settings, where interaction histories are assumed to be linear and individually maintained.
As a result, they are not well suited to the multi-threaded, collaborative, and interdependent dynamics of multi-agent systems.

\paragraph{LLM-based multi-agent systems.}

%
The deployment of multiple LLM agents has enabled collaborative capabilities beyond single-agent systems.
Early frameworks, including AutoGen~\citep{wu2024autogen}, CAMEL~\citep{Li2023CAMELCA}, and AgentVerse~\citep{chen2023agentverse}, established multi-agent collaboration through predefined roles, fixed interaction patterns, or manually designed coordination topologies.
MetaGPT~\citep{Hong2023MetaGPTMP} and ChatDev~\citep{qian2024chatdev} further formalized role-based workflows, while DyLAN~\citep{liu2023dynamic} introduced inference-time agent selection and Agent Importance Score for adaptive collaboration.
Recent methods such as GPTSwarm~\citep{zhuge2024gptswarm}, AFlow~\citep{zhang2024aflow}, AgentNet~\citep{Yang2025AgentNetDE}, and Mixture-of-Minds~\citep{zhou2025mixture} further automate the design or optimization of coordination structures.
However, these approaches still focus primarily on task-level or benchmark-level coordination design, rather than persistent memory-driven improvement across tasks. 
As a result, most existing systems remain limited to static or one-shot adaptation and do not support genuine self-evolution through learning from past collaborations over time.

\paragraph{Memory in multi-agent systems.}
The intersection of memory and multi-agent systems remains relatively underexplored~\citep{Zhang2024ASO,Yu2026MultiAgentMF}.
Many existing MAS frameworks either omit dedicated memory components or incorporate only rudimentary within-trial memory, such as retaining the immediate conversation history~\citep{wu2024autogen, liu2023dynamic}. 
Existing cross-session memory designs often store condensed task-level artifacts rather than rich collaboration trajectories.
For example, ChatDev mainly preserves distilled task summaries or past solutions, while discarding most fine-grained inter-agent interaction details~\citep{qian2024chatdev}. 
MetaGPT employs a global shared message pool as a centralized communication and memory substrate~\citep{Hong2023MetaGPTMP}, and G-Memory further extends this line of work with a hierarchical graph-based shared memory that stores insights, queries, and interaction traces~\citep{zhang2025g}. 
However, such centralized memory repositories introduce synchronization and consistency overhead, access-control and privacy risks, and scalability bottlenecks~\citep{Rezazadeh2025CollaborativeMM, Yu2026MultiAgentMF}.
Moreover, exposing all agents to the same shared memory may homogenize agent behavior and weaken the emergence of specialized, role-complementary expertise~\citep{Yu2026MultiAgentMF, Yang2025AgentNetDE}.
These limitations motivate decentralized memory architectures in which each agent maintains its own private experience while selectively sharing useful knowledge~\citep{Rezazadeh2025CollaborativeMM, Yu2026MultiAgentMF, Yang2025AgentNetDE}.

\section{Preliminary}\label{sec:prelim}
We first introduce the notation and formalize the basic concepts of a multi-agent system.
Let \(\mathcal{X}=\{x_1,\dots,x_J\}\) denote the set of tasks, and let \(\mathcal{A}=\{a_1,\dots,a_M\}\) denote the set of collaborative agents for solving tasks in \(\mathcal{X}\). Each agent \(a_m \in \mathcal{A}\) is characterized by
\(
a_m = \bigl(\mathrm{Base}_m,\mathrm{Role}_m,\mathcal{M}_m,\mathrm{Tool}_m\bigr),
\)
where \(\mathrm{Base}_m\) denotes the underlying LLM, \(\mathrm{Role}_m\) denotes the functional role or persona of the agent, \(\mathcal{M}_m=\{z_1,z_2,\dots,z_j\}\) denotes its memory space, and \(\mathrm{Tool}_m\) denotes the set of tools accessible to the agent. 
Each memory piece \(z_i \in \mathcal{M}_m\) is represented as \(z_i = (\xi_i, r_i^\star),\) where \(\xi_i\) is the context prototype and \(r_i^\star\) is the associated action prototype.
Depending on the memory framework, \(\mathcal{M}_m\) may refer either to memory retrieved from a centralized repository or to the private memory space maintained by agent \(a_m\) in a decentralized setting.

Given an input task \(x \in \mathcal{X}\), the system evolves through a finite sequence of stages
\(
\mathcal{T}=\{t_1,\dots,t_N\},
\)
where each stage corresponds to one step of collaborative problem solving, such as decomposition, reasoning, verification, or integration.
At stage \(t\), let \(\mathcal{A}^{(t)} \subseteq \mathcal{A}\) denote the set of active agents. Multiple agents may be active simultaneously at the same stage.
For each active agent \(a_m \in \mathcal{A}^{(t)}\), we define its local context as
\(
c_m^{(t)} = \bigl(x,~\phi_m^{(t)},~\rho_m^{(t)}\bigr),
\)
where \(\phi_m^{(t)}\) denotes the information induced by previous actions of neighboring or predecessor agents, and \(\rho_m^{(t)}\) denotes the memory retrieved from \(\mathcal{M}_m\). 
Conditioned on the local context \(c_m^{(t)}\), agent \(a_m\) produces an action or intermediate output
\(
v_m^{(t)}=\pi_m\bigl(c_m^{(t)}\bigr),
\)
where \(\pi_m\) denotes the policy of agent \(a_m\). 
After all active agents produce their outputs, the system aggregates them as
\(
s^{(t)}=\mathrm{Agg}\bigl(\{v_m^{(t)}\}_{a_m\in\mathcal{A}^{(t)}}\bigr),
\)
and outputs the final answer.
\section{\ours}\label{sec:method}
\begin{figure}[t]
    \centering
    \includegraphics[width=\linewidth]{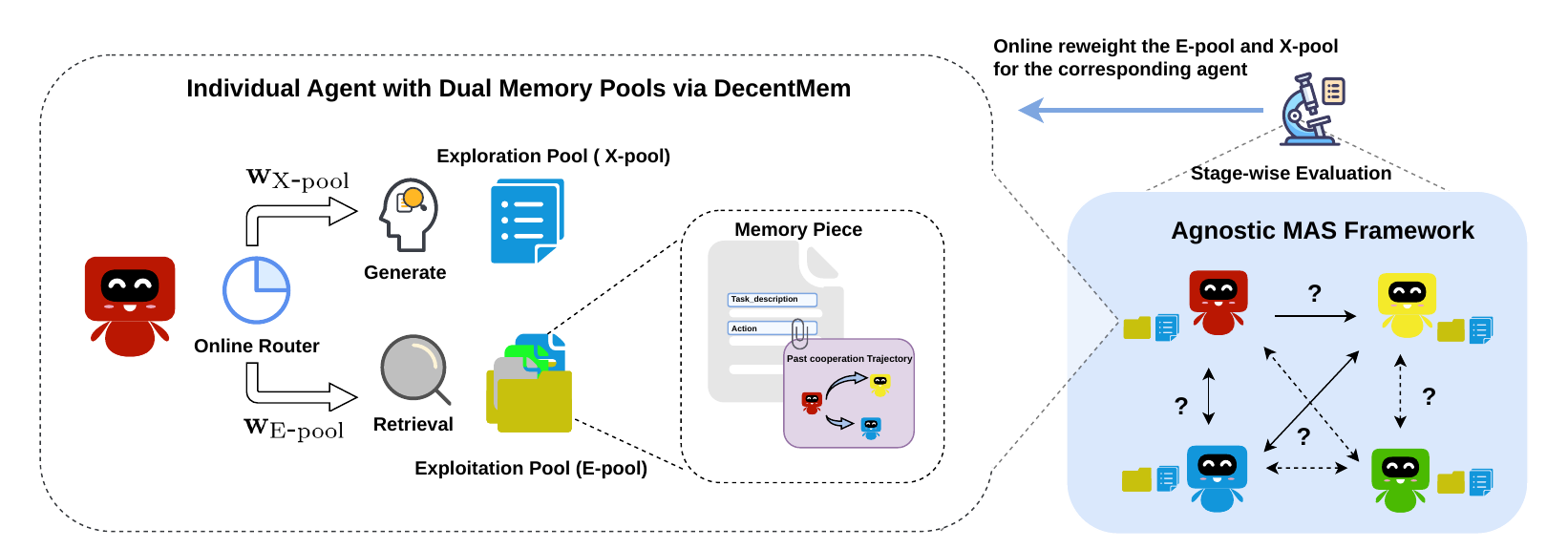}
    \vspace{-0.25in}
    \caption{
        Overview of \ours.
        Each agent maintains a private dual-pool memory (left): an \textit{exploitation pool (E-pool)} of consolidated trajectories from past tasks and an \textit{exploration pool (X-pool)} for generating novel candidates in unseen contexts. 
        At each stage, an online router selects between the two pools with probability proportional to their weights $w_{E-pool}$ and $w_{X-pool}$, retrieving from the E-pool or generating from the X-pool to produce the agent's action. 
        \ours is agnostic to the underlying MAS framework (right): agents collaborate under any topology, from pre-designed workflows to opportunistic coordination. 
        After execution, an external LLM-as-a-judge evaluates the trajectory stage by stage and reweights each agent's E-pool and X-pool based on stage-wise feedback.
        %
        }
    \label{fig:example}
    \vspace{-0.15in}
\end{figure}

In this section, we present the key design of \ours, a \textit{decentralized} memory mechanism for multi-agent systems. 
Each agent is equipped with dual memory pools that store collaboration-aware memory pieces. 
At execution time, an online router retrieves task-relevant experience from one pool or generates new memory from the other for exploration; after execution, the two pools are reweighted online based on stage-wise feedback from an LLM-as-a-judge.
%

\subsection{Dual-Memory Pools}
For each agent \(a_m\), \ours partitions the accessible memory into two pools:
\begin{equation}
    \mathcal{M}_m
    =
    \mathcal{M}_{m,\mathrm{E\mbox{-}pool}}
    \cup
    \mathcal{M}_{m,\mathrm{X\mbox{-}pool}}.
\end{equation}

The exploitation pool (E-pool), denoted by \(\mathcal{M}_{m,\mathrm{E\mbox{-}pool}}\), stores consolidated experience from previous tasks and supports retrieval-based exploitation over historical memory pieces that encode past collaboration trajectories. 
In contrast, the exploration pool (X-pool) serves as a temporary buffer for the current task and is used to generate novel experiences. 
As will be shown in \secref{sec:theory}, this dual-memory design helps the system achieve global reachability over the solution space.
Specifically, the E-pool enables retrieval-based exploitation through a local-walk mechanism, whereas the X-pool realizes a heuristic teleportation mechanism induced by the LLM prior.
Together, the two pools define a mixed search process that balances local exploitation with global exploration while preventing the search from being trapped in the local optima.

Continue with the definitions from \secref{sec:prelim}, each memory piece in \(\mathcal{M}_{m,\mathrm{E\mbox{-}pool}} \) is represented as
\[
z=\bigl(\xi,r^\star\bigr),
\qquad
r^\star = \big( r_{\mathrm{trajectory}}, r_{\mathrm{comment}}\big)
\]
where \( r_{\mathrm{trajectory}}\) includes specific action such as task decomposition and direct answering, together with cooperation trajectory passed to subsequent agents for the next stage induced by the action, and \(r_{\mathrm{comment}}\) denotes the agent’s action-level self-commentary, including its rationale for choosing the corresponding action in past interactions.
Unlike standard MAS memory, our memory piece $z$ records not only \emph{what} was solved, but also \emph{how} the task was solved and \emph{who} executed each sub-task.
Retrieved memory therefore serves as both an action prior and a coordination-trajectory prior, making each memory piece valuable not only for retrieval and policy reuse, but also for directly learning from past successful cooperation experience. 
%

\subsection{Memory Retrieval for New Tasks}
\ours is designed as a plug-in module that can be integrated into mainstream MAS frameworks.
When a task \(x\) arrives, it is first dispatched to an active agent \(a_n\) according to the underlying scheduling policy. 
Suppose that at stage \(t\), an active agent \(a_m\) encounters an intermediate sub-task \(x^{(t)}\). 
The online router of \(a_m\) first selects one memory pool from the dual-memory structure according to the corresponding weights \(w_{m,\mathrm{E\mbox{-}pool}}\) and \(w_{m,\mathrm{X\mbox{-}pool}}\). The probability of selecting the E-pool is defined as
\begin{equation}
    \alpha_{m,\mathrm{E\mbox{-}pool}}
    =
    \frac{w_{m,\mathrm{E\mbox{-}pool}}}
    {w_{m,\mathrm{E\mbox{-}pool}}+w_{m,\mathrm{X\mbox{-}pool}}},
    \qquad
    \alpha_{m,\mathrm{X\mbox{-}pool}}
    =
    1-\alpha_{m,\mathrm{E\mbox{-}pool}}.
\end{equation}
If the E-pool is selected, the agent retrieves memory pieces according to similarity:
\begin{equation}
\rho_m^{(t)}
=
\operatorname{Top\mbox{-}K}\!\left(
\mathcal{M}_{m,\mathrm{E\mbox{-}pool}};
\ \mathrm{sim}\bigl(e(x^{(t)}),e(\xi_z)\bigr)
\right),
\end{equation}
subject to a similarity threshold \(\tau\). 
If no sufficiently similar memory is found, the online router falls back to the X-pool.
This threshold avoids mismatched reuse from the E-pool when the current task deviates substantially from previously consolidated experience, thereby promoting faster adaptation to novel tasks.
If the X-pool is selected directly, the agent instantiates an exploratory memory piece \(z_{\text{new}} \) for the current context.
Either the retrieved set \(\rho_m^{(t)}\) or exploratory memory piece \(z_{\text{new}} \) is then provided to the LLM for  executable action generation:
\begin{equation}
v_m^{(t)}
=
\Pi_{\mathrm{act}}\!\bigl(x^{(t)},\phi_m^{(t)},\rho_m^{(t)} / z_{\text{new}}\bigr).
\end{equation}
Notably, when a retrieved E-pool memory piece contains a compatible collaboration trajectory, its historical allocation can be reused as a coordination prior for the next stage. 
After all active agents produce their outputs, the system aggregates them as
\begin{equation}
s^{(t)}
=
\mathrm{Agg}\bigl(\{v_m^{(t)}\}_{a_m\in\mathcal{A}^{(t)}}\bigr),
\end{equation}
and outputs the final answer after the last stage.

\subsection{Dual-Memory Update}
\label{sec:dual}
After execution, the full solution trajectory is evaluated stage by stage by an LLM evaluator. 
Rather than scoring only the final answer, the evaluator assesses each stage in terms of \textit{correctness}, \textit{allocation quality}, \textit{intermediate coherence}, and \textit{final integration}.
Specifically, let \(q_{\mathrm{prev}}\) and \(q_{\mathrm{curr}}\) denote the scores of two consecutive stages, and define
\(
\Delta_t=\mathbb{I}[q_{\mathrm{curr}} > q_{\mathrm{prev}}].
\)
The online router is then updated according to the memory pool used at stage \(t\).
Let \(\alpha=0.5\) and \(\beta=0.5\), if the E-pool is used, its weight is increased when \(\Delta_t=1\), and otherwise decayed:
\begin{equation}
w_{m,\mathrm{E\mbox{-}pool}}
\leftarrow
\begin{cases}
w_{m,\mathrm{E\mbox{-}pool}}+\alpha, & \text{if } \Delta_t=1,\\[4pt]
\max(1.0,\beta\,w_{m,\mathrm{E\mbox{-}pool}}), & \text{otherwise}.
\end{cases}
\end{equation}
In contrast, if the X-pool is used, the update is reversed:
\begin{equation}
w_{m,\mathrm{E\mbox{-}pool}}
\leftarrow
\begin{cases}
\max(1.0,\beta\,w_{m,\mathrm{E\mbox{-}pool}}), & \text{if } \Delta_t=1,\\[4pt]
w_{m,\mathrm{E\mbox{-}pool}}+\alpha, & \text{otherwise}.
\end{cases}
\end{equation}

Meanwhile, \(w_{m,\mathrm{X\mbox{-}pool}}=1.0\) remains fixed. 
In this way, successful exploitation increases reliance on the E-pool, while successful exploration prevents the router from over-committing to past experience.

After the task is completed, all memory pieces in the X-pool are consolidated into the E-pool, and the X-pool is reset:
\begin{equation}
\mathcal{M}_{m,\mathrm{E\mbox{-}pool}}
\leftarrow
\mathcal{M}_{m,\mathrm{E\mbox{-}pool}} \cup \mathcal{M}_{m,\mathrm{X\mbox{-}pool}},
\qquad
\mathcal{M}_{m,\mathrm{X\mbox{-}pool}} \leftarrow \varnothing.
\end{equation}
\section{Theoretical Analysis}
\label{sec:theory}
In this section, we formalize self-evolving multi-agent systems as strategy search over a graph-structured solution space and cast the learning objective as cumulative regret minimization. 
We then show that \ours guarantees global reachability over the decentralized search space and achieves \(O(\log T)\) cumulative regret.

\subsection{Problem Formulation}
We model the full solution space as a graph
\(
\mathcal{G}=(V,\mathcal{E}),
\)
where each node \(v\in V\) denotes a candidate solution strategy, and each edge \((v,v')\in\mathcal{E}\) indicates that the two strategies are highly similar and can be transformed into one another.
For each agent \(a_m\), we define an agent-specific solution subspace
\[
\mathcal{G}_m=(V_m,\mathcal{E}_m),
\qquad
V_m\subseteq V,\;\mathcal{E}_m\subseteq\mathcal{E},
\qquad
\mathcal{G}=\mathcal{G}_1\cup\mathcal{G}_2\cup\cdots\cup\mathcal{G}_M.
\]
Thus, \(\mathcal{G}_m\) characterizes the local geometry of the strategy space accessible to agent \(a_m\). Unlike centralized memory, which searches directly over the global space \(\mathcal{G}\), \ours enables each agent to search within its own subspace \(\mathcal{G}_m\).
Given a task \(x\), let
\(
\tau(x)=\{v_m^{(t)} \mid t=1,\dots,N,\ a_m\in\mathcal{A}^{(t)}\}
\)
denote the joint trajectory induced by all active agents, \(R_x\bigl(\tau(x),\hat{y}(x)\bigr)\) the task reward, which evaluates both the collaborative process and the final output, and
\(
R_x^\star=\max_{\tau,\hat{y}} R_x(\tau,\hat{y})
\)
the optimal achievable reward for task \(x\), the cumulative regret over the task set \(\mathcal{X}\) and the objective are therefore
\[
\mathrm{Regret}(\mathcal{X})
=
\sum_{x\in\mathcal{X}}
\left(
R_x^\star-\mathbb{E}\!\left[R_x\bigl(\tau(x),\hat{y}(x)\bigr)\right]
\right),
\quad
\mathrm{Objective}=\min_{\{\pi_m\}_{m=1}^M}\mathrm{Regret}(\mathcal{X}).
\]
%

\subsection{Global Reachability}
We analyze the search behavior of each agent as a dynamic random walk over its graph-structured solution subspace. 
For a given agent \(a_m\), let
\(
\mathcal{Q}_m=\{q_{m,1},q_{m,2},\dots,q_{m,L_m}\}
\)
denote the sequence of local subproblems encountered by \(a_m\) during execution, where \(q_{m,\ell}\) is the \(\ell\)-th subproblem assigned to agent \(a_m\). 
Note this sequence is agent-specific and should be distinguished from the global stage sequence used to describe the overall multi-agent collaboration process.

Let \(p_{m,\ell}\in\mathbb{R}^{|V_m|}\) denote the state distribution of agent \(a_m\) when solving subproblem \(q_{m,\ell}\), where \(p_{m,\ell}\) lies in the probability simplex over \(V_m\). 
In \ours, the two memory pools induce two complementary search operators.
The exploitation pool (E-pool) defines a local walk governed by a similarity-based transition matrix \(T_m\), which captures the reuse of historically relevant strategies. 
The exploration pool (X-pool) defines a teleportation mechanism induced by an LLM prior \(h_m\in\Delta^{|V_m|}\), where \((h_m)_i>0\) for every feasible state \(i\in V_m\).
The resulting transition is
\begin{equation}
p_{m,\ell+1}
=
M_{m,\ell}p_{m,\ell}
=
\left(
\alpha_{m,\ell}T_m+(1-\alpha_{m,\ell})\,h_m\mathbf{1}^\top
\right)p_{m,\ell},
\label{eq:global_transition}
\end{equation}
where
\(
\alpha_{m,\ell}
=
\frac{w_{m,\mathrm{E\mbox{-}pool}}}
{w_{m,\mathrm{E\mbox{-}pool}}+w_{m,\mathrm{X\mbox{-}pool}}}.
\label{eq:alpha_t}
\)
\eqnref{eq:global_transition} shows that \ours combines local exploitation with global exploration within the local subspace \(\mathcal{G}_m\): the first term preserves local structural bias, while the second term injects nonzero probability mass into every feasible region.

\begin{theorem}[Global Reachability]
Assume that \(w_{m,\mathrm{X\mbox{-}pool}} > 0\) throughout the search process, so that \(\alpha_{m,\ell}<1\) for all \(\ell\), and assume that \((h_m)_i>0\) for every feasible state \(i\in V_m\). Then the transition matrix
\[
M_{m,\ell}
=
\alpha_{m,\ell}T_m+(1-\alpha_{m,\ell})\,h_m\mathbf{1}^\top
\]
is strictly positive. Consequently, the induced Markov chain over \(V_m\) is irreducible and aperiodic, and the search process of agent \(a_m\) is globally reachable over its entire subspace \(\mathcal{G}_m\).
\end{theorem}

Since this argument holds for every agent \(a_m\), \ours ensures global reachability across the decentralized solution space \(\mathcal{G}=\bigcup_{m=1}^M \mathcal{G}_m\). 
The detailed proof is deferred to \appref{app:global_reachability}.

\subsection{Logarithmic Cumulative Regret}
We next analyze the routing efficiency of \ours from an online optimization perspective. 
At round \(t\), the online router of agent \(a_m\) selects between the E-pool and the X-pool according to their current weights. 
Let \(r_{\mathrm{E}}(\alpha)\) and \(r_{\mathrm{X}}(\alpha)\) denote the expected rewards of the E-pool and X-pool under routing probability \(\alpha\), respectively. 
The resulting expected reward is
\(
r(\alpha)
=
\alpha\,r_{\mathrm{E}}(\alpha)
+
(1-\alpha)\,r_{\mathrm{X}}(\alpha).
\label{eq:r_alpha_main}
\)

\begin{Assumption}
\label{assumption function r}
The function \(r:[0.5,1]\to[0,1]\) is strictly concave, twice continuously differentiable, and admits a unique maximizer \(\alpha^\star\in(0.5,1)\).
\end{Assumption}

For a routing sequence \(\{\alpha_t\}_{t=1}^T\), we define the cumulative regret as
\(
R(T)
=
\sum_{t=1}^T
\Bigl(
r(\alpha^\star)-\mathbb{E}[r(\alpha_t)]
\Bigr).
\label{eq:regret_main}
\)
The update rule of \ours induces an adaptive recursion on \(\alpha_t\): successful exploitation increases \(\alpha_t\), whereas successful exploration suppresses over-dominance of the E-pool and decreases \(\alpha_t\). 
The router steers the system toward the optimal exploitation--exploration balance.

\begin{theorem}[Logarithmic Regret]
Under Assumption~\ref{assumption function r}, the routing policy of \ours satisfies
\(
\mathbb{E}[R(T)] = O(\log T).
\)
Hence, \ours achieves \(O(\log T)\) cumulative regret.
\end{theorem}
The proof is detailed in \appref{app:log_regret}. 
Moreover, stochastic bandit theory implies that, for some problem instances, any policy incurs expected cumulative regret at least \(\Omega(\log T)\)~\citep{Auer2002FinitetimeAO}.
Therefore, the regret bound of \ours is order-optimal up to constant factors.
\section{Experiment}\label{sec:experiments}
\subsection{Experimental Setup}
\paragraph{Datasets.}
We evaluate \ours on five public benchmarks spanning four task categories. 
These cover mathematical reasoning (AIME25~\citep{math-ai2025}, AIME24~\citep{Maxwell-Jia2024}), code generation (MBPP-Plus~\citep{evalplus_mbppplus_hf}), question answering (BBH~\citep{suzgun2023challenging}), and embodied decision-making (ALFWorld~\citep{shridhar2020alfworld}). AIME25 and AIME24 are combined into a single  dataset and reported as AIME25\&24 in the tables.

\paragraph{Baselines.}
We compare \ours against three centralized memory baselines including MetaGPT~\citep{Hong2023MetaGPTMP}, ChatDev~\citep{qian2024chatdev}, and G-Memory~\citep{zhang2025g}, and a no-memory baseline that runs the underlying MAS framework without any cross-session memory.

\paragraph{MAS Frameworks and LLM Backbones.} 
A central design property of \ours is that it is agnostic to both the underlying MAS framework and the LLM backbone. 
To validate this, we deliberately span the experimental matrix along both axes.
We integrate \ours and all baselines into three MAS frameworks chosen to cover a broad spectrum of coordination structures: \textbf{AutoGen}~\citep{wu2024autogen}, which represents pre-designed workflow-based collaboration; \textbf{DyLAN}~\citep{liu2023dynamic}, which captures more dynamic and partially stochastic interaction patterns; and \textbf{AgentNet}~\citep{Yang2025AgentNetDE}, which reflects unstructured and opportunistic coordination. 
The degree of stochasticity increases progressively from AutoGen to AgentNet, allowing us to test whether \ours's gains are robust across coordination regimes. 
And, as we will show in \secref{subsec:main_results}, the gains in fact widen as stochasticity grows.
We instantiate each MAS framework with five open-source LLM backbones spanning two model families and various sizes: \textsc{Qwen3-4B}, \textsc{Qwen3-8B}, \textsc{Qwen3-14B}~\citep{Yang2025Qwen3TR}, \textsc{Gemma4-E2B-it}, and \textsc{Gemma4-E4B-it}~\citep{farabet2026gemma4}.
This range lets us verify that \ours's benefits do not depend on a particular model family, scale, or architecture. 
We instantiate the \textsc{Qwen} series models locally via Ollama~\citep{ollama}, and the \textsc{Gemma} series models locally using the Hugging Face Transformers library~\citep{transformers}.
%


\subsection{Main Results}\label{subsec:main_results}
\paragraph{\ours consistently outperforms all baselines.} 
As shown in \tabref{tab:main_qwen3} and~\ref{tab:main_gemma4}, \ours achieves the best average accuracy in 14 of 15 (backbone, framework) cells, with an average relative improvement of 8.6\% over the strongest centralized baseline and 26.1\% over the no-memory baseline (up to 52.5\% on Qwen3-4B + AgentNet). 
Beyond steady-state accuracy, \ours also exhibits stronger self-evolution with accumulated experience, reaching strong performance with substantially fewer tasks (\secref{subsec:evolve}).

\paragraph{The advantage of decentralization grows with coordination stochasticity.} 
As stochasticity increases from AutoGen (pre-designed workflows) to DyLAN (partially stochastic) to AgentNet (opportunistic), the relative margin widens monotonically on both families: 2.7\% / 9.2\% / 23.1\% on Qwen3 and 1.7\% / 3.9\% / 6.8\% on Gemma4. 
This is direct evidence for the paper's central claim --- when coordination is less predictable, centralized memory homogenizes agent behavior more aggressively, and preserving per-agent experience matters more.
\begin{table}[t]
  \centering
  \caption{Main results on the Qwen3 family across three MAS frameworks and five datasets. \best{Bold} marks the best result and \underline{underline} marks the second-best within each (backbone, framework) block.}
  \label{tab:main_qwen3}
  \setlength{\tabcolsep}{10pt}
  \resizebox{\textwidth}{!}{%
    \begin{tabular}{c c c | c c c c c}
    \toprule
    \textbf{Backbone} & \textbf{Framework} & \textbf{Memory} & \textbf{AIME25\&24} & \textbf{MBPP-Plus} & \textbf{ALFWorld} & \textbf{BBH} & \textbf{Avg.} \\
    \midrule

    \multirow{15}{*}{\rotatebox{90}{\textsc{Qwen3-4B}}}
     & \multirow{5}{*}{AgentNet}
       & No memory   & 13.33 & 53.20 & 46.12 & 28.13 & 35.20 \\
     & & MetaGPT     & 8.33  & 59.14 & 50.47 & 40.86 & 39.70 \\
     & & ChatDev     & \underline{18.33} & 56.40 & 49.23 & 33.67 & 39.41 \\
     & & G-Memory    & 10.00 & \underline{63.18} & \underline{52.14} & \underline{48.14} & \underline{43.37} \\
     & & \textbf{\ours} & \best{23.33} & \best{67.53} & \best{64.45} & \best{59.43} & \best{53.69} \\
    \cmidrule(l){2-8}
     & \multirow{5}{*}{DyLAN}
       & No memory   & 16.67 & 54.72 & 48.54 & 41.17 & 40.28 \\
     & & MetaGPT     & 15.00 & 58.60 & 51.17 & \underline{45.23} & 42.50 \\
     & & ChatDev     & \underline{21.67} & 61.23 & 48.56 & 44.73 & 44.05 \\
     & & G-Memory    & 18.33 & \underline{63.93} & \underline{52.30} & 43.97 & \underline{44.63} \\
     & & \textbf{\ours} & \best{26.67} & \best{64.50} & \best{57.25} & \best{52.19} & \best{50.15} \\
    \cmidrule(l){2-8}
     & \multirow{5}{*}{AutoGen}
       & No memory   & 18.33 & 58.14 & 54.27 & 44.43 & 43.79 \\
     & & MetaGPT     & 11.67 & 58.92 & 59.53 & 54.19 & 46.08 \\
     & & ChatDev     & \underline{20.00} & \underline{64.72} & 63.87 & 53.37 & 50.49 \\
     & & G-Memory    & 15.00 & 63.87 & \best{69.17} & \underline{61.53} & \underline{52.39} \\
     & & \textbf{\ours} & \best{21.67} & \best{65.63} & \underline{67.73} & \best{63.80} & \best{54.71} \\
    \midrule

    \multirow{15}{*}{\rotatebox{90}{\textsc{Qwen3-8B}}}
     & \multirow{5}{*}{AgentNet}
       & No memory   & 23.33 & 62.94 & 61.40 & 35.13 & 45.70 \\
     & & MetaGPT     & 13.34 & 66.94 & 66.51 & 57.14 & 50.98 \\
     & & ChatDev     & \underline{31.67} & 64.99 & 64.03 & 42.11 & 50.70 \\
     & & G-Memory    & 16.67 & \underline{70.94} & \underline{70.19} & \underline{62.18} & \underline{55.00} \\
     & & \textbf{\ours} & \best{36.67} & \best{77.93} & \best{80.29} & \best{76.19} & \best{67.77} \\
    \cmidrule(l){2-8}
     & \multirow{5}{*}{DyLAN}
       & No memory   & 26.67 & 65.91 & 63.19 & 54.68 & 52.61 \\
     & & MetaGPT     & 23.33 & 67.97 & 68.21 & \underline{61.58} & 55.27 \\
     & & ChatDev     & \underline{31.67} & 70.94 & 64.64 & 59.61 & 56.72 \\
     & & G-Memory    & 30.00 & \best{74.49} & \underline{69.43} & 58.62 & \underline{58.14} \\
     & & \textbf{\ours} & \best{40.00} & \underline{72.99} & \best{74.42} & \best{65.52} & \best{63.23} \\
    \cmidrule(l){2-8}
     & \multirow{5}{*}{AutoGen}
       & No memory   & 28.33 & 66.94 & 78.81 & 61.88 & 58.99 \\
     & & MetaGPT     & 20.00 & 68.99 & 83.25 & 76.12 & 62.09 \\
     & & ChatDev     & \underline{28.33} & 71.05 & 89.14 & 71.14 & 64.91 \\
     & & G-Memory    & 25.00 & \underline{73.92} & \best{92.11} & \underline{82.09} & \underline{68.28} \\
     & & \textbf{\ours} & \best{31.67} & \best{75.97} & \underline{91.54} & \best{85.07} & \best{71.06} \\
    \midrule

    \multirow{15}{*}{\rotatebox{90}{\textsc{Qwen3-14B}}}
     & \multirow{5}{*}{AgentNet}
       & No memory   & 26.67 & 68.72 & 66.20 & 42.61 & 51.05 \\
     & & MetaGPT     & 20.00 & 68.81 & 71.32 & 65.48 & 56.40 \\
     & & ChatDev     & \underline{36.67} & \underline{76.78} & 68.50 & 48.92 & 57.72 \\
     & & G-Memory    & 23.33 & 71.29 & \underline{75.61} & \underline{68.52} & \underline{59.69} \\
     & & \textbf{\ours} & \best{41.66} & \best{82.27} & \best{83.59} & \best{84.30} & \best{72.95} \\
    \cmidrule(l){2-8}
     & \multirow{5}{*}{DyLAN}
       & No memory   & 30.00 & 72.10 & 68.82 & 60.52 & 57.86 \\
     & & MetaGPT     & 26.67 & 76.95 & 73.11 & \underline{68.86} & 61.40 \\
     & & ChatDev     & \underline{35.00} & 79.17 & 69.43 & 65.28 & 62.22 \\
     & & G-Memory    & 31.67 & \best{85.75} & \underline{74.24} & 66.84 & \underline{64.62} \\
     & & \textbf{\ours} & \best{40.00} & \underline{82.39} & \best{80.11} & \best{72.43} & \best{68.73} \\
    \cmidrule(l){2-8}
     & \multirow{5}{*}{AutoGen}
       & No memory   & 30.00 & 74.21 & 81.51 & 67.33 & 63.26 \\
     & & MetaGPT     & 23.33 & 76.37 & 89.43 & 81.68 & 67.70 \\
     & & ChatDev     & \underline{33.33} & 80.36 & \underline{94.89} & 79.28 & 71.97 \\
     & & G-Memory    & 33.33 & \best{85.23} & \best{96.69} & \underline{89.20} & \best{76.11} \\
     & & \textbf{\ours} & \best{35.00} & \underline{84.74} & 93.13 & \best{90.50} & \underline{75.84} \\
    \bottomrule
    \end{tabular}%
  }
  \vspace{-0.25in}
\end{table}
\begin{table}[t]
  \centering
  \caption{Main results on the Gemma4 family across three MAS frameworks and five datasets. \best{Bold} marks the best result and \underline{underline} marks the second-best within each (backbone, framework) block.}
  \label{tab:main_gemma4}
  \setlength{\tabcolsep}{10pt}
  \resizebox{\textwidth}{!}{%
    \begin{tabular}{c c c | c c c c c}
    \toprule
    \textbf{Backbone} & \textbf{Framework} & \textbf{Memory} & \textbf{AIME25\&24} & \textbf{MBPP-Plus} & \textbf{ALFWorld} & \textbf{BBH} & \textbf{Avg.} \\
    \midrule

    \multirow{15}{*}{\rotatebox{90}{\textsc{Gemma4-E2B}}}
     & \multirow{5}{*}{AgentNet}
       & No memory   & 23.33 & 55.67 & 60.43 & 21.32 & 40.19 \\
     & & MetaGPT     & 18.33 & 57.99 & 63.37 & 30.85 & 42.64 \\
     & & ChatDev     & \underline{26.67} & 57.00 & 62.27 & 27.11 & 43.26 \\
     & & G-Memory    & 20.00 & \underline{62.13} & \underline{68.47} & \underline{38.15} & \underline{47.19} \\
     & & \textbf{\ours} & \best{25.00} & \best{64.08} & \best{70.80} & \best{42.56} & \best{50.61} \\
    \cmidrule(l){2-8}
     & \multirow{5}{*}{DyLAN}
       & No memory   & 23.33 & 57.62 & 62.29 & 28.67 & 42.98 \\
     & & MetaGPT     & 21.66 & 60.15 & 64.51 & 35.42 & 45.44 \\
     & & ChatDev     & \underline{28.33} & 62.81 & 63.00 & 31.47 & 46.40 \\
     & & G-Memory    & 25.00 & \best{69.89} & \underline{68.11} & \underline{36.39} & \underline{49.85} \\
     & & \textbf{\ours} & \best{28.33} & \underline{66.79} & \best{69.53} & \best{43.47} & \best{52.03} \\
    \cmidrule(l){2-8}
     & \multirow{5}{*}{AutoGen}
       & No memory   & 25.00 & 59.77 & 66.38 & 35.15 & 46.58 \\
     & & MetaGPT     & 28.33 & 62.93 & 67.49 & 40.83 & 49.89 \\
     & & ChatDev     & \underline{30.00} & \underline{65.57} & \best{74.59} & \best{46.45} & \underline{54.15} \\
     & & G-Memory    & 28.33 & 64.82 & 66.73 & \underline{42.17} & 50.51 \\
     & & \textbf{\ours} & \best{33.33} & \best{69.28} & \underline{73.43} & 41.87 & \best{54.48} \\
    \midrule

    \multirow{15}{*}{\rotatebox{90}{\textsc{Gemma4-E4B}}}
     & \multirow{5}{*}{AgentNet}
       & No memory   & 33.33 & 70.23 & 75.67 & 31.92 & 52.79 \\
     & & MetaGPT     & 25.00 & 74.49 & 77.47 & 48.14 & 56.28 \\
     & & ChatDev     & \underline{38.33} & 72.18 & \underline{81.59} & 35.74 & 56.96 \\
     & & G-Memory    & 30.00 & \best{78.39} & \best{82.39} & \underline{55.38} & \underline{61.54} \\
     & & \textbf{\ours} & \best{43.33} & \underline{77.17} & 81.46 & \best{59.69} & \best{65.41} \\
    \cmidrule(l){2-8}
     & \multirow{5}{*}{DyLAN}
       & No memory   & 36.67 & 72.49 & 76.28 & 45.17 & 57.65 \\
     & & MetaGPT     & 31.67 & 74.69 & 80.49 & 52.41 & 59.82 \\
     & & ChatDev     & \underline{38.33} & 75.37 & 77.23 & 50.69 & 60.41 \\
     & & G-Memory    & 36.67 & \best{82.45} & \best{84.49} & \underline{55.83} & \underline{64.86} \\
     & & \textbf{\ours} & \best{41.67} & \underline{81.49} & \underline{84.21} & \best{60.84} & \best{67.05} \\
    \cmidrule(l){2-8}
     & \multirow{5}{*}{AutoGen}
       & No memory   & 35.00 & 75.52 & 82.73 & 49.53 & 60.70 \\
     & & MetaGPT     & 36.67 & 77.69 & 86.27 & 58.37 & 64.75 \\
     & & ChatDev     & \best{41.67} & \best{83.04} & 88.99 & 55.46 & 67.29 \\
     & & G-Memory    & 38.33 & \underline{81.17} & \underline{89.17} & \underline{63.72} & \underline{68.10} \\
     & & \textbf{\ours} & \underline{40.00} & 79.39 & \best{92.28} & \best{68.09} & \best{69.94} \\
    \bottomrule
    \end{tabular}%
}
\end{table}

\section{Analysis and Ablation Study}
\subsection{Self-Evolution with Experience}\label{subsec:evolve}
%
%
\figref{fig:experience} traces cumulative accuracy on MBPP-Plus across the three MAS frameworks under Qwen3-4B. 
\ours reaches strong performance substantially faster than all three centralized baselines, with the convergence-speed advantage most pronounced on DyLAN (\textasciitilde2.5× faster). 
On AgentNet, \ours also surpasses every baseline's final accuracy, indicating that under stochastic coordination it not only learns faster but discovers strategies centralized memory cannot recover.

\begin{figure}[h!]
    \centering
    \includegraphics[width=1\linewidth]{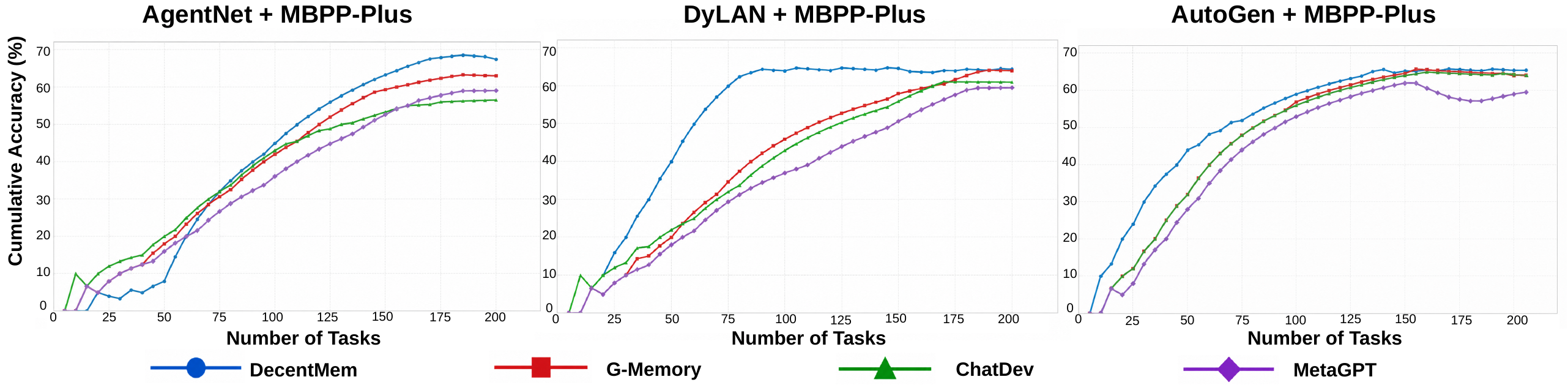}
    \vspace{-0.2in}
    \caption{
        Cumulative accuracy vs.\ number of tasks on MBPP-Plus across three MAS frameworks under Qwen3-4B.
        %
        %
        \ours exhibits stronger self-evolution than all three centralized baselines.
    }
    \label{fig:experience}
    \vspace{-0.2in}
\end{figure}
\begin{figure}[htb]
    \centering
    \includegraphics[width=.9\linewidth]{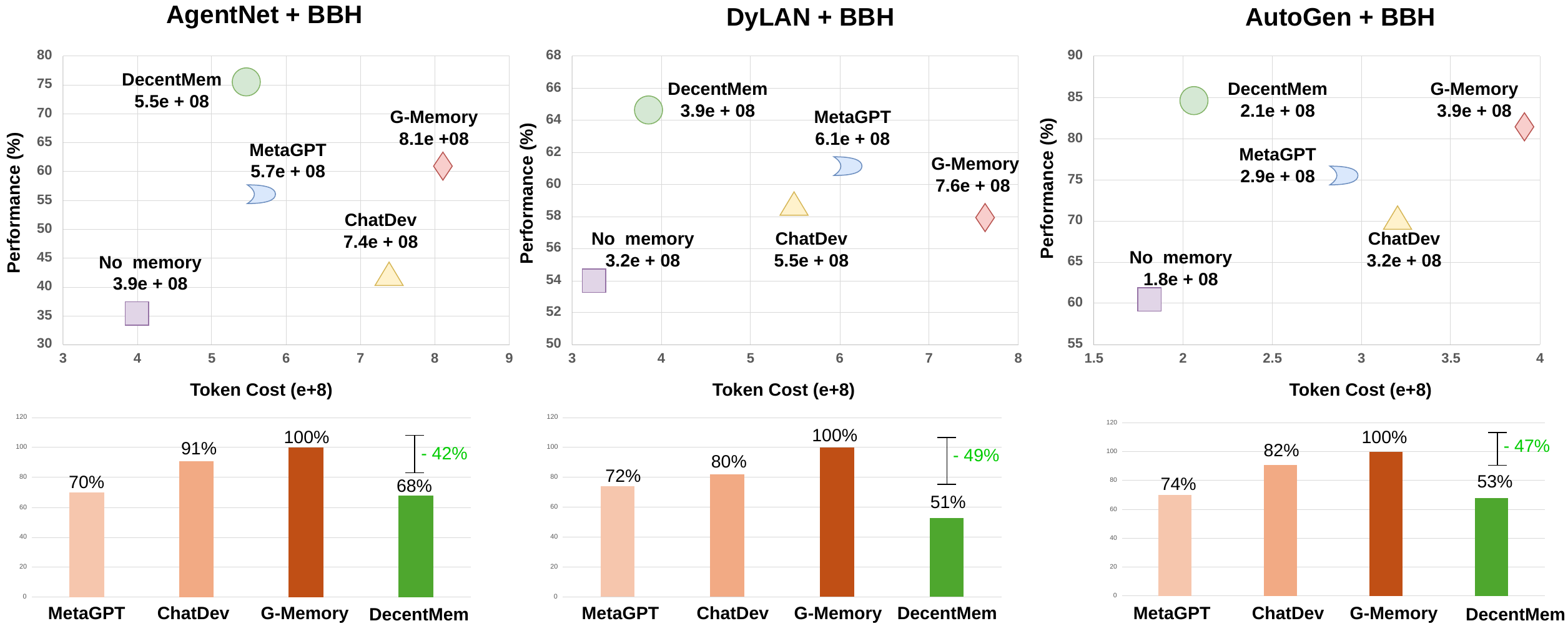}
    \vspace{-0.05in}
    \caption{
        Token cost vs.\ performance on BBH across three MAS frameworks under QWEN3-8B.
    }
    \label{fig:cost_analysis}
    \vspace{-0.15in}
\end{figure}

\begin{table}[h!]
\centering
\caption{
    Ablation of the online router. 
    All variants share the same dual-pool architecture; only the routing policy differs. 
    \textbf{Bold} indicates the best performance.
}
\label{tab:ablation_router}
\footnotesize
\setlength{\tabcolsep}{4pt}
\renewcommand{\arraystretch}{0.95}
\resizebox{\textwidth}{!}{%
\begin{tabular}{@{}ll cc cc cc cc cc cc@{}}
\toprule
 & & \multicolumn{4}{c}{\textbf{AgentNet}}
   & \multicolumn{4}{c}{\textbf{DyLAN}}
   & \multicolumn{4}{c}{\textbf{AutoGen}} \\
\cmidrule(lr){3-6}\cmidrule(lr){7-10}\cmidrule(lr){11-14}
 & & \multicolumn{2}{c}{Qwen3-8B} & \multicolumn{2}{c}{Gemma4-E4B}
   & \multicolumn{2}{c}{Qwen3-8B} & \multicolumn{2}{c}{Gemma4-E4B}
   & \multicolumn{2}{c}{Qwen3-8B} & \multicolumn{2}{c}{Gemma4-E4B} \\
\cmidrule(lr){3-4}\cmidrule(lr){5-6}\cmidrule(lr){7-8}\cmidrule(lr){9-10}\cmidrule(lr){11-12}\cmidrule(lr){13-14}
\textbf{Routing} & $\alpha$
 & BBH & AIME & BBH & AIME
 & BBH & AIME & BBH & AIME
 & BBH & AIME & BBH & AIME \\
\midrule
Exploit & 1   & 74.27 & 26.67 & 57.23 & 30.00
              & 63.59 & 31.67 & 58.72 & 40.00
              & 84.29 & 28.33 & 65.37 & 33.33 \\
Explore & 0   & 35.13 & 23.33 & 31.92 & 33.33
              & 54.68 & 26.67 & 45.17 & 36.67
              & 61.88 & 28.33 & 49.53 & 35.00 \\
Fixed   & 0.5 & 65.28 & 33.33 & 53.18 & 40.00
              & 60.21 & 35.00 & 55.38 & \textbf{43.33}
              & 78.34 & \textbf{31.67} & 63.38 & 38.33 \\
Online  & --
 & \textbf{76.19} & \textbf{36.67} & \textbf{59.69} & \textbf{43.33}
 & \textbf{65.52} & \textbf{40.00} & \textbf{60.84} & 41.67
 & \textbf{85.07} & \textbf{31.67} & \textbf{68.09} & \textbf{40.00} \\
\bottomrule
\end{tabular}}
\vspace{-0.2in}
\end{table}

\subsection{Cost Analysis}\label{subsec:cost}
Beyond accuracy, a practical memory framework must scale efficiently in token usage.
\figref{fig:cost_analysis} plots performance against token cost on BBH across the three MAS frameworks under Qwen3-8B. 
\ours sits at the upper-left corner of all three plots --- highest accuracy at the lowest token cost. 
Relative to G-Memory, the strongest centralized baseline, \ours reduces token consumption by 32\%, 49\%, and 47\% on AgentNet, DyLAN, and AutoGen respectively (43\% on average) while improving accuracy.

\subsection{Ablation Study}\label{subsec:ablation}
To isolate the contribution of the online weighting, we compare it against three fixed routing strategies sharing the same dual-memory architecture: \textit{Exploitation Only} ($\alpha=1$), \textit{Exploration Only} ($\alpha=0$), and \textit{Fixed Weight} ($\alpha=0.5$). 
As shown in \tabref{tab:ablation_router}, replacing it with the strongest fixed policy (\textit{Exploitation Only}) costs 6.93\%, 3.51\%, and 3.38\% on AgentNet, DyLAN, and AutoGen respectively, with larger gaps against the other two. 
Notably, \textit{Exploration Only} collapses on BBH but remains competitive on AIME, showing that exploration alone suffices for fresh reasoning problems but cannot substitute for accumulated experience on tasks with reusable structure.

\newpage
\section{Conclusion and Limitation}

In this paper, we argued that centralized memory is the non-optimal default for multi-agent systems, as sharing a single pool erodes the role-complementary specialization that motivates having multiple agents at all, and proposed \ours, a decentralized dual-pool memory with online routing between exploitation and exploration. 
The construction admits an $O(\log T)$ regret guarantee under a graph-walk-with-teleportation reading, and the empirical margin over the strongest centralized baseline grows monotonically as coordination becomes more stochastic.
%


\paragraph{Limitation.} 
Although \ours is evaluated across four domains, broader validation on more diverse and high-stakes tasks, such as legal and medical reasoning, would further strengthen its empirical soundness. 
We leave this direction for future work.

\bibliographystyle{unsrtnat}
\bibliography{references}

\newpage
\appendix

\section{Additional Theoretical Analysis}
\label{app:theory}

In this appendix, we provide a rigorous theoretical foundation for the dual-memory design of \textbf{\ours}. By modeling each agent's search process as a random walk with heuristic teleportation over a graph-structured solution subspace, we show that the adaptive dual-memory mechanism guarantees global reachability and yields superior search efficiency in terms of asymptotic regret.

\subsection{Global Reachability}
\label{app:global_reachability}

We formalize the search process of each agent as a random walk over its graph-structured solution subspace. Since agents make decisions independently conditioned on their local context, it suffices to analyze a single agent. For agent \(a_m\), let
\[
\mathcal{G}_m=(V_m,\mathcal{E}_m)
\]
denote its local solution graph, where each node in \(V_m\) represents a candidate solution strategy and each edge in \(\mathcal{E}_m\) encodes a local similarity relation between strategies.

During execution, agent \(a_m\) encounters a sequence of local subproblems
\[
\mathcal{Q}_m=\{q_{m,1},q_{m,2},\dots,q_{m,L_m}\},
\]
where \(q_{m,\ell}\) denotes the \(\ell\)-th subproblem assigned to \(a_m\). Let
\[
p_{m,\ell}\in\Delta^{|V_m|}
\]
denote the state distribution of agent \(a_m\) when solving \(q_{m,\ell}\), where \(\Delta^{|V_m|}\) is the probability simplex over \(V_m\).

In \textbf{\ours}, the two memory pools induce two complementary search operators. The exploitation pool (E-pool) defines a local walk governed by a column-stochastic transition matrix
\[
T_m\in\mathbb{R}^{|V_m|\times |V_m|},
\]
where \((T_m)_{ij}\) denotes the probability of moving from state \(j\) to state \(i\) according to historical similarity. The exploration pool (X-pool) defines a teleportation mechanism induced by an LLM prior
\[
h_m\in\Delta^{|V_m|},
\qquad
(h_m)_i>0,\ \forall i\in V_m.
\]
The resulting transition when moving from subproblem \(q_{m,\ell}\) to \(q_{m,\ell+1}\) is
\begin{equation}
p_{m,\ell+1}
=
M_{m,\ell}p_{m,\ell}
=
\left(
\alpha_{m,\ell}T_m
+
(1-\alpha_{m,\ell})\,h_m\mathbf{1}^{\top}
\right)p_{m,\ell},
\label{eq:appendix_global_transition}
\end{equation}
where
\begin{equation}
\alpha_{m,\ell}
=
\frac{w_{m,\mathrm{E\mbox{-}pool},\ell}}
{w_{m,\mathrm{E\mbox{-}pool},\ell}+w_{m,\mathrm{X\mbox{-}pool},\ell}}
\in [0,1).
\label{eq:appendix_alpha}
\end{equation}

Equation~\eqref{eq:appendix_global_transition} shows that \textbf{\ours} combines local exploitation with global exploration. The term \(\alpha_{m,\ell}T_m\) preserves the local geometry induced by past experience, while the rank-one term \((1-\alpha_{m,\ell})h_m\mathbf{1}^{\top}\) injects nonzero probability mass into every feasible state.

\begin{theorem}[Global Reachability]
\label{thm:global_reachability}
Assume that the exploration pool remains active throughout the search process, i.e.,
\[
w_{m,\mathrm{X\mbox{-}pool},\ell}>0
\qquad
\text{for all } \ell,
\]
so that \(\alpha_{m,\ell}<1\), and assume that the LLM prior has full support over the feasible state space:
\[
(h_m)_i>0,
\qquad
\forall i\in V_m.
\]
Then the transition matrix
\[
M_{m,\ell}
=
\alpha_{m,\ell}T_m
+
(1-\alpha_{m,\ell})h_m\mathbf{1}^{\top}
\]
is strictly positive. Consequently, the induced Markov chain is irreducible and aperiodic. In particular, from any initial state, the search process can reach any region of the solution subspace with nonzero probability.
\end{theorem}

\begin{proof}
For any pair of states \((i,j)\), the \((i,j)\)-th entry of \(M_{m,\ell}\) is
\begin{equation}
(M_{m,\ell})_{ij}
=
\alpha_{m,\ell}(T_m)_{ij}
+
(1-\alpha_{m,\ell})(h_m)_i.
\label{eq:appendix_entry}
\end{equation}
Since \((T_m)_{ij}\ge 0\), \(\alpha_{m,\ell}<1\), and \((h_m)_i>0\), we have
\[
(M_{m,\ell})_{ij}>0,
\qquad
\forall i,j.
\]
Hence \(M_{m,\ell}\) is a strictly positive stochastic matrix.

Any strictly positive stochastic matrix is irreducible and aperiodic. Therefore, the Markov chain induced by \(M_{m,\ell}\) admits a unique stationary distribution with strictly positive support on all states. Equivalently, for any two states in \(V_m\), there exists a path of nonzero probability connecting them under the dynamics of \textbf{\ours}. Thus, the search process cannot be permanently trapped in a local absorbing region and remains globally reachable over the entire solution subspace \(\mathcal{G}_m\).
\end{proof}

\begin{remark}[Boundary Cases]
The condition \(\alpha_{m,\ell}\in(0,1)\) is essential.

\textbf{Pure exploitation (\(\alpha_{m,\ell}=1\)).}
In this case, the transition reduces to
\[
M_{m,\ell}=T_m.
\]
If \(T_m\) contains a closed communicating class corresponding to a locally consistent but suboptimal region, then the chain is reducible. Once the process enters that region, the escape probability is zero, and global reachability is lost.

\textbf{Pure exploration (\(\alpha_{m,\ell}=0\)).}
In this case, the transition reduces to
\[
M_{m,\ell}=h_m\mathbf{1}^{\top},
\qquad
p_{m,\ell+1}=h_m.
\]
The chain remains ergodic, but the dynamics become memoryless: the next state no longer depends on the current state or the local geometry encoded by \(T_m\). Hence the process loses the ability to accumulate and refine historical experience.

These two boundary cases clarify why \textbf{\ours} requires an adaptive mixture of the two operators: E-pool alone may get trapped locally, whereas X-pool alone discards structured reuse. Their combination yields both global reachability and meaningful self-evolution.
\end{remark}

\subsection{Superior Search Efficiency}
\label{app:log_regret}

We now establish the search efficiency of \textbf{\ours}. We show that its online routing mechanism achieves logarithmic cumulative regret, which is order-optimal, and asymptotically outperforms any fixed routing policy.

\begin{theorem}[Logarithmic Regret]
\label{theorem: reg}
Under Assumption\ref{assumption function r}, the routing policy of \textbf{\ours} satisfies
\[
\mathbb{E}[R(T)] = O(\log T).
\]
Hence, \textbf{\ours} achieves \(O(\log T)\) cumulative regret.
\end{theorem}

\begin{proof}
We first analyze the online router introduced in Section~\ref{sec:dual}. For notational simplicity, we suppress the agent index \(m\) when no ambiguity arises. Recall that
\begin{equation}
\alpha_{\ell}
=
\frac{w_{\ell}}{w_{\ell}+1},
\label{eq:appendix_alpha_def}
\end{equation}
where
\[
w_{\ell}:=w_{m,\mathrm{E\mbox{-}pool},\ell},
\qquad
w_{m,\mathrm{X\mbox{-}pool},\ell}=1.
\]
Thus, \(\alpha_{\ell}\in[0.5,1)\) is the probability of selecting the E-pool when solving the \(\ell\)-th local subproblem.

At local step \(\ell\), the router samples
\[
I_{\ell}\in\{0,1\},
\]
where \(I_{\ell}=1\) denotes selecting the E-pool and \(I_{\ell}=0\) denotes selecting the X-pool. The routing probabilities are
\[
\mathbb{P}(I_{\ell}=1\mid\mathcal{F}_{\ell})=\alpha_{\ell},
\qquad
\mathbb{P}(I_{\ell}=0\mid\mathcal{F}_{\ell})=1-\alpha_{\ell},
\]
where \(\mathcal{F}_{\ell}\) denotes the filtration generated by the routing history up to step \(\ell\). After the routing decision, the system observes a binary reward
\[
R_{\ell}\in\{0,1\},
\]
indicating whether the stage-wise score improves.

Let \(r_{\mathrm{E}}(\alpha)\) and \(r_{\mathrm{X}}(\alpha)\) denote the expected rewards of the E-pool and X-pool under routing probability \(\alpha\), respectively. The induced expected reward is
\begin{equation}
r(\alpha)
=
\alpha\,r_{\mathrm{E}}(\alpha)
+
(1-\alpha)\,r_{\mathrm{X}}(\alpha),
\qquad
\alpha\in[0.5,1).
\label{eq:appendix_reward_alpha}
\end{equation}
Under Assumption~\ref{assumption function r}, \(r\) is strictly concave and admits a unique maximizer
\[
\alpha^\star\in(0.5,1).
\]

\paragraph{One-step recursion.}
Ignoring the projection \(\max(1,\cdot)\) for asymptotic analysis, the update rule in Section~4.3 becomes
\begin{equation}
w_{\ell+1}
=
\begin{cases}
w_{\ell}+\frac{1}{2},
& \text{if } (I_{\ell}=1,R_{\ell}=1)\ \text{or}\ (I_{\ell}=0,R_{\ell}=0),\\[4pt]
\frac{1}{2}w_{\ell},
& \text{if } (I_{\ell}=1,R_{\ell}=0)\ \text{or}\ (I_{\ell}=0,R_{\ell}=1).
\end{cases}
\label{eq:appendix_w_update}
\end{equation}
Using \eqref{eq:appendix_alpha_def}, we obtain
\begin{equation}
\alpha_{\ell+1}
=
\begin{cases}
\phi_{+}(\alpha_{\ell}) := \dfrac{1+\alpha_{\ell}}{3-\alpha_{\ell}},
& \text{if } (I_{\ell}=1,R_{\ell}=1)\ \text{or}\ (I_{\ell}=0,R_{\ell}=0),\\[10pt]
\phi_{-}(\alpha_{\ell}) := \dfrac{\alpha_{\ell}}{2-\alpha_{\ell}},
& \text{if } (I_{\ell}=1,R_{\ell}=0)\ \text{or}\ (I_{\ell}=0,R_{\ell}=1).
\end{cases}
\label{eq:appendix_alpha_update}
\end{equation}

Define
\begin{equation}
q(\alpha)
:=
\alpha\,r_{\mathrm{E}}(\alpha)
+
(1-\alpha)\bigl(1-r_{\mathrm{X}}(\alpha)\bigr),
\label{eq:appendix_q_alpha}
\end{equation}
which is the probability of taking the first branch in \eqref{eq:appendix_alpha_update}. Then
\begin{equation}
\mathbb{E}[\alpha_{\ell+1}-\alpha_{\ell}\mid\mathcal{F}_{\ell}]
=
g(\alpha_{\ell}),
\label{eq:appendix_drift}
\end{equation}
where
\begin{equation}
g(\alpha)
=
q(\alpha)\!\left(\frac{1+\alpha}{3-\alpha}-\alpha\right)
+
\bigl(1-q(\alpha)\bigr)\!\left(\frac{\alpha}{2-\alpha}-\alpha\right).
\label{eq:appendix_g_alpha}
\end{equation}

Since \(r\) is strictly concave with unique maximizer \(\alpha^\star\), the induced mean dynamics are locally contractive around \(\alpha^\star\). In particular, there exists \(\lambda>0\) such that
\begin{equation}
g(\alpha)(\alpha-\alpha^\star)
\le
-\lambda(\alpha-\alpha^\star)^2
\label{eq:appendix_contraction}
\end{equation}
for all \(\alpha\) in a neighborhood of \(\alpha^\star\).

\paragraph{Stochastic approximation.}
After standard time rescaling, \eqref{eq:appendix_alpha_update} can be written in Robbins--Monro form:
\begin{equation}
\alpha_{\ell+1}
=
\alpha_{\ell}
+
\frac{1}{\ell}g(\alpha_{\ell})
+
\frac{1}{\ell}\xi_{\ell+1},
\label{eq:appendix_rm}
\end{equation}
where \(\ell\) denotes the update scaling parameter, and \(\{\xi_{\ell}\}\) is a martingale difference sequence satisfying
\[
\mathbb{E}[\xi_{\ell+1}\mid\mathcal{F}_{\ell}] = 0,
\qquad
\sup_{\ell}\mathbb{E}\!\left[\|\xi_{\ell+1}\|^2\mid\mathcal{F}_{\ell}\right] < \infty
\quad \text{a.s.}
\]
By standard stochastic approximation results \citep[Theorem~2.2]{borkar2024stochastic},
\begin{equation}
\mathbb{E}\bigl[(\alpha_{\ell}-\alpha^\star)^2\bigr]
=
O(1/\ell).
\label{eq:appendix_mse}
\end{equation}

By strong concavity of \(r\), there exists \(\mu>0\) such that
\begin{equation}
r(\alpha^\star)-r(\alpha)
\le
\mu(\alpha-\alpha^\star)^2,
\qquad
\forall \alpha\in[0.5,1).
\label{eq:appendix_concavity}
\end{equation}
Therefore,
\begin{equation}
\mathbb{E}[R(T)]
\le
\mu\sum_{\ell=1}^T
\mathbb{E}\bigl[(\alpha_{\ell}-\alpha^\star)^2\bigr].
\label{eq:appendix_regret_reduce}
\end{equation}
Substituting \eqref{eq:appendix_mse} into \eqref{eq:appendix_regret_reduce} yields
\[
\mathbb{E}[R(T)]
\le
\mu\sum_{\ell=1}^T O(1/\ell)
=
O(\log T).
\]
\end{proof}

The cumulative regret over \(T\) local decisions is
\begin{equation}
R(T)
=
\sum_{\ell=1}^T
\Bigl(
r(\alpha^\star)-\mathbb{E}[r(\alpha_{\ell})]
\Bigr).
\label{eq:appendix_regret_alpha}
\end{equation}
Since stochastic bandit lower bounds imply that \(\Omega(\log T)\) regret is unavoidable in general \citep{Auer2002FinitetimeAO}, an \(O(\log T)\) upper bound is order-optimal.

\begin{corollary}[Fixed Routing Is Suboptimal]
For any fixed routing probability \(\bar{\alpha}\neq\alpha^\star\), the cumulative regret is linear:
\[
R_{\mathrm{fixed}}(T)=\Theta(T).
\]
Consequently, \textbf{\ours} is asymptotically superior to any fixed routing policy.
\end{corollary}

\begin{proof}
Let \(\bar{\alpha}\neq\alpha^\star\) be fixed. Since \(r\) is strictly concave and attains its unique maximum at \(\alpha^\star\), the per-step gap
\[
\delta:=r(\alpha^\star)-r(\bar{\alpha})
\]
is a strictly positive constant. Therefore,
\[
R_{\mathrm{fixed}}(T)
=
\sum_{\ell=1}^T \delta
=
\delta T
=
\Theta(T).
\]
By Theorem~\ref{theorem: reg}, \textbf{\ours} satisfies
\[
\mathbb{E}[R(T)] = O(\log T).
\]
Hence,
\[
\lim_{T\to\infty}
\frac{\mathbb{E}[R(T)]}{R_{\mathrm{fixed}}(T)}
=0,
\]
which implies the claim.
\end{proof}

\begin{remark}
In particular, the balanced fixed routing choice \(\bar{\alpha}=0.5\), which corresponds to assigning equal selection probability to the two memory pools at every step, is a special case of a fixed policy and is therefore asymptotically dominated by \textbf{\ours}.
\end{remark}

\section{Experiment Details}
\subsection{Baseline Setup}

In this section, we provide detailed descriptions of the centralized memory baselines used in our comparisons.

\textbf{MetaGPT.}
This memory design is derived from MetaGPT~\citep{Hong2023MetaGPTMP} and focuses exclusively on within-trial memory, i.e., information stored and shared internally during the resolution of a single task by multiple agents.

\textbf{ChatDev.}
This memory design is adapted from ChatDev~\citep{qian2024chatdev}, which incorporates both within-trial and cross-trial memory. The within-trial memory is passed from the central or initiating agent at the beginning of each round to guide subsequent interactions based on prior context. The cross-trial memory is relatively simple, storing past solutions to previous queries for future retrieval.

\textbf{G-Memory.}
This memory design is directly adopted from G-Memory~\citep{zhang2025g}, which manages long multi-agent interaction histories through a three-tier graph hierarchy consisting of insight, query, and interaction graphs.

\subsection{Multi-agent system setup}

In this section, we detail the setups of our three adopted MAS framework, AgentNet, DyLAN and AutoGen.

\textbf{AgentNet.}
AgentNet~\cite{Yang2025AgentNetDE} is a decentralized multi-agent framework characterized by unstructured and opportunistic coordination, where collaboration is not governed by a fixed workflow or a central controller, but instead emerges dynamically from local agent decisions. Specifically, each agent is able to decide whether to forward a task, decompose it into subtasks, or execute it locally based on the current task state and its own capability. It then performs the assigned subtask and produces task-specific outputs. This flexible structure enables scalable collaboration without relying on predefined interaction patterns. In our implementation of AgentNet, three agents participate in a three-stage collaboration process.

\textbf{DyLAN.}
DyLAN~\cite{liu2023dynamic} is a debate-style framework similar to LLM-Debate, but it incorporates a more efficient agent-wise early-stopping mechanism during multi-turn interactions. DyLAN further employs an agent selection algorithm based on an unsupervised metric, namely the Agent Importance Score, to identify the most contributive agents through a preliminary trial tailored to the target task. In our implementation of DyLAN, three agents participate in the debate, while an additional ranker agent evaluates their relative importance.

\textbf{AutoGen.}
AutoGen~\cite{wu2024autogen} is a widely used multi-agent orchestration framework for coordinating interactions among specialized agents in problem-solving tasks. Specifically, we adopt its A3: Decision Making structure, which consists of: (1) a Solver Agent, responsible for generating solutions and initialized with the system prompt ``You are a smart agent designed to solve problems.''; (2) a Ground Truth Agent, which critically evaluates the solver's output and identifies potential errors against a reference standard; and (3) an Executor Agent, which translates validated solutions into executable commands. This modular design enables transparent, verifiable, and actionable multi-agent collaboration.

\section{Additional Experiment Result}

\subsection{Case Study}

In this section, we use G-Memory as a representative centralized memory baseline on the AutoGen framework to explain why centralized memory incurs higher token cost. The main overhead comes from three stages: memory retrieval, memory reading, and memory update.

In centralized memory, all agents access a shared memory repository. When a new task arrives, each active agent retrieves relevant memory from the global memory space, which contains memory pieces accumulated by the entire multi-agent system. Therefore, even if only a few memory pieces are finally selected, the retrieval process still operates over a large and heterogeneous repository containing experiences from different agents and roles. In contrast, \ours decentralizes memory at the agent level. Each agent maintains its own private dual-memory pools, consisting of an exploitation pool and an exploration pool. During memory retrieval, an agent only accesses its own local memory space, which is much smaller and more role-specific. This reduces irrelevant retrieval and avoids repeatedly exposing agents to memory pieces generated by other agents.

The same difference also appears in memory reading and update. In centralized memory, agents may repeatedly read overlapping global memories, and after task completion the system needs to summarize and update the shared repository. In \ours, each agent reads only its personalized memory and updates only its own dual-memory pools. Thus, no system-wide memory synchronization or global memory rewriting is required.

Overall, the high token cost of centralized memory mainly comes from its global access pattern: agents retrieve, read, and update memory against the whole shared repository. By localizing these operations to agent-private dual-memory pools, \ours reduces token consumption while preserving role-specific experience.

\newtcolorbox{decentcasebox}[1]{
    enhanced,
    breakable,
    colback=white,
    colframe=decentBorder,
    colbacktitle=decentTitle,
    coltitle=black,
    fonttitle=\bfseries\large,
    title={#1},
    boxrule=0.8pt,
    arc=2pt,
    left=6pt,
    right=6pt,
    top=6pt,
    bottom=6pt,
    toptitle=6pt,
    bottomtitle=6pt,
    lefttitle=8pt,
    righttitle=8pt,
    before skip=10pt,
    after skip=12pt,
    borderline west={2pt}{0pt}{decentBorder}
}

\newcommand{\casebar}[1]{%
    \par\vspace{0.45em}
    \noindent
    \colorbox{decentBar}{%
        \parbox{\dimexpr\linewidth-2\fboxsep\relax}{%
            \bfseries\small\color{caseText} #1%
        }%
    }%
    \par\vspace{0.25em}
}

\newenvironment{casecontent}{
    \par\noindent
    \begingroup
    \leftskip=5pt
    \rightskip=5pt
    \vspace{4pt}
    \small
    \setlength{\parindent}{0pt}
    \setlength{\parskip}{3pt}
}{
    \par
    \vspace{4pt}
    \endgroup
}

\DefineVerbatimEnvironment{greenlog}{Verbatim}{
    fontsize=\scriptsize,
    breaklines=true,
    breakanywhere=true,
    commandchars=\\\{\},
    frame=single,
    framesep=4pt,
    rulecolor=\color{decentBorder},
    formatcom=\color{decentLogText}
}

\begin{decentcasebox}{Case Study (DecentMem)}

\casebar{Task}
\begin{casecontent}
It is not always easy to see which chemicals are contained in our consumer products.
The following argument pertains to this question. First, vitamin A is an ingredient of
LIQUID EYELINER. Second, every ingredient of BC Eye Cream 15 ml is not an ingredient
of Mango Lip Butter or not an ingredient of LIQUID EYELINER. Therefore, it is not the
case that vitamin A is an ingredient of BC Eye Cream 15 ml.

Is the argument, given the explicitly stated premises, deductively valid or invalid?

\vspace{0.4em}
\textbf{[Task Manager]} \\
Now I have task received: determine whether the chemical-ingredient argument is
deductively valid or invalid.

\vspace{0.4em}
\textbf{[Online router | Dual Memory Pool Access]} \\
Memory pattern: DecentMem

\vspace{0.4em}
Query: \\
``BBH deductive validity; natural-language logic; ingredient relation; disjunction;
valid/invalid option format''
\end{casecontent}

\casebar{Memory Search Process}
\begin{casecontent}
\begin{greenlog}
Accessed pool:
agent_id = 0
stage_pool = agent_0_stage_pool["1"]
E_pool_weight = 1.50
X_pool_weight = 1.00

Online router: E-pool
Candidate local memories:
AP-f1-old-09 | sim=0.74 | prior task: "check whether conclusion follows only from stated premises"
AP-f1-old-12 | sim=0.46 | prior task: "multiple-choice answer must be one of the given options"
AP-f1-new-03 | sim=0.59 | recent note: "avoid importing external product knowledge"

Selected compact memory packet for Solver Agent 0:
\end{greenlog}
\end{casecontent}

\casebar{Memory Reading Process}
\begin{casecontent}
\begin{greenlog}
MemoryFragment(
  fragment_id = "agent_0_L1_mem_171642_old_49",
  stage = "1",
  pool_type = "old",
  similarity_to_current_task = 0.74,

  environment = {
    "problem": "Determine whether the argument is deductively valid or invalid.",
    "task_type": "bbh/chemical_logical_deduction",
    "answer_space": ["valid", "invalid"]
  },

  rule_definition = RuleDefinition(
    rule_name = "C1",
    description = "Initial-stage agent for checking whether the conclusion follows from the stated premises."
  ),

  action = ActionRecord(
    action_type = "direct_answer",
    direct_answer = "invalid",
    thought_summary = "Use only stated premises. If a counterexample exists, answer invalid."
  ),

  task_allocation = [],
  result_quality = 9.0,
  created = 18
)

prompt_injected_to_solver:
"Past experience: For deductive validity tasks, separate the premises from the conclusion,
use only the stated premises, and answer invalid if a counterexample is possible."
\end{greenlog}

\vspace{0.4em}
\textbf{[\#1] Solver Agent 0} \\
I receive the task plus the compact local memory packet.

\textbf{Output:} \\
invalid

......
\end{casecontent}

\casebar{Memory Update Process}
\begin{casecontent}
\begin{greenlog}
Agent Pool Local Write | After L1
Write a compact stage result to agent_0_stage_pool["1"].E_pool:

memory_type = "deductive_validity_short_trace"
action = "initial_answer"
answer = "invalid"
support = "uses premise"

Agent Pool Weight Update | L1
Update only the local L1 pool. The retrieved strategy receives a higher local weight because it supported a high-scoring answer. No centralized memory merge is triggered.

Agent Dual Pool Memory Update:
Store compact local traces in the assigned agents' stage-specific pools.
\end{greenlog}
\end{casecontent}

\end{decentcasebox}

\newcommand{\gmemsection}[1]{%
    \par\vspace{0.45em}
    \noindent
    \colorbox{gmemBar}{%
        \parbox{\dimexpr\linewidth-2\fboxsep\relax}{%
            \bfseries\small\color{caseText} #1%
        }%
    }%
    \par\vspace{0.25em}
}

\newenvironment{redlog}{%
  \VerbatimEnvironment
  \begingroup
  \color{gmemLogText}%
  \renewcommand{\FancyVerbFormatLine}[1]{\textcolor{gmemLogText}{##1}}%
  \begin{Verbatim}[
    fontsize=\scriptsize,
    breaklines=true,
    breakanywhere=true,
    frame=single,
    framesep=4pt,
    rulecolor=\color{gmemBorder},
    fillcolor=\color{gmemLogBg},
    formatcom=\color{gmemLogText}
  ]%
}{%
  \end{Verbatim}%
  \endgroup
}

\DefineVerbatimEnvironment{blacklog}{Verbatim}{
    fontsize=\scriptsize,
    breaklines=true,
    breakanywhere=true
}

\begin{gmemorycasebox}{Case Study (G-Memory)}

\gmemsection{Task}
\begin{gmemcontent}
It is not always easy to see which chemicals are contained in our consumer products.
The following argument pertains to this question. First, vitamin A is an ingredient of
LIQUID EYELINER. Second, every ingredient of BC Eye Cream 15 ml is not an ingredient
of Mango Lip Butter or not an ingredient of LIQUID EYELINER. Therefore, it is not the
case that vitamin A is an ingredient of BC Eye Cream 15 ml.

\vspace{0.4em}
Is the argument, given the explicitly stated premises, deductively valid or invalid?

\vspace{0.6em}
\textbf{[Task Manager]} \\
Now BBH task received: determine whether the chemical-ingredient argument is
deductively valid or invalid.

\vspace{0.6em}
\textbf{[G-Memory Controller | Centralized Memory Pool Access]} \\
Memory pattern: G-Memory.

\vspace{0.6em}
\textbf{Query:} \\
``BBH deductive validity; natural-language logic; ingredient relation; disjunction;
valid/invalid option format''
\end{gmemcontent}

\gmemsection{Memory Search Process}
\begin{gmemcontent}
\begin{redlog}
Initialize task context:
task_id = bbh_logical_deduction_case
framework = AutoGen
stage = [solver, coordinator, executor]
answer_space = {valid, invalid}

Create centralized working state:
current_task_node
stage_message_buffer
retrieved_success_pool
retrieved_failure_pool
retrieved_insight_pool
agent_message_graph

[G-memory Query Builder]
Build a higher query from the whole task:

semantic keywords = "deductive validity, ingredient relation, universal quantifier, disjunction, countermodel"
structural pattern = "premise P(x); premise forall x, B(x)->not M(x) or not P(x); conclusion not B(a)"
expected reasoning tool = "counterexample / model construction"

[G-memory Retrieval | Search Multiple Central Pools]
Search spaces:
successful_trajectory_pool
failed_trajectory_pool
task_insight_pool
evaluator_feedback_pool
agent_message_graph

Retrieved successful memories:
GM-S-017 | sim=0.88 | weight=1.42 | "For validity tasks, construct a countermodel if the conclusion is not forced."
GM-S-042 | sim=0.84 | weight=1.31 | "Disjunction in a universal premise leaves multiple satisfying branches."
GM-S-068 | sim=0.80 | weight=1.18 | "Translate natural-language relations into predicates before deciding."

Retrieved failed memories:
GM-F-011 | sim=0.78 | weight=1.25 | failure: "Mistook A -> (not B or not C) as directly implying not B."
GM-F-028 | sim=0.73 | weight=1.10 | failure: "Imported external commonsense instead of using explicitly stated premises."

Retrieved task-level insights:
GM-I-008 | sim=0.88 | "A single satisfying countermodel is enough to prove invalidity."
GM-I-014 | sim=0.77 | "Final answer must be exactly one of the provided options."

[G-memory Maintenance Pass | Before Solver]
The centralized memory manager performs extra work:

1. Normalize retrieved records into a common schema.
2. Deduplicate near-identical memories.
3. Convert failed trajectories into negative constraints.
4. Re-rank memories by similarity, success weight, recency, and stage relevance.
5. Compress the selected records into a solver-specific memory packet.

Memory packet injected to Solver Agent 0:

Positive guidance:
Formalize the premises with predicates.
Try to construct a countermodel.
If one model makes all premises true and the conclusion false, answer invalid.

Negative guidance:
Do not treat "not M(x) or not E(x)" as "not B(x)".
Do not use outside facts about cosmetics or ingredients.
\end{redlog}
\end{gmemcontent}

\gmemsection{Memory Reading Process}
\begin{gmemcontent}
\begin{redlog}
CentralMemoryRecord(
  memory_id = "GM-S-017",
  pool_type = "successful_trajectory_pool",
  memory_type = "successful_task_trajectory",
  similarity_to_current_task = 0.88,
  retrieval_weight = 1.42,
  source_framework = "AutoGen",
  source_task_family = "BBH deductive validity",

  task_signature = {
    "input_modality": "natural_language_argument",
    "answer_format": "binary valid/invalid",
    "logic_pattern": "universal conditional with disjunction",
    "recommended_tool": "predicate formalization + countermodel construction"
  },

  successful_trajectory = {
    "L1_solver": {
      "role": "solver",
      "action": "formalize premises",
      "intermediate_result": "identify predicates and map the conclusion to a target formula",
      "quality_score": 9.0
    },
    "L2_ground_truth": {
      "role": "validator",
      "action": "construct countermodel",
      "intermediate_result": "find an assignment where all premises hold and the conclusion fails",
      "quality_score": 10.0
    }
  },

  prompt_injection_text =
  "Use predicate formalization and countermodel search. A single model satisfying the premises
  and falsifying the conclusion is sufficient to prove invalidity."
)

CentralMemoryRecord(
  memory_id = "GM-F-011",
  pool_type = "failed_trajectory_pool",
  memory_type = "failure_case_with_correction",
  similarity_to_current_task = 0.79,
  retrieval_weight = 1.25,

  failed_task_signature = {
    "logic_pattern": "conditional premise containing disjunction",
    "surface_error": "agent treated a weak disjunctive condition as a strong exclusion",
    "observed_wrong_answer": "valid",
    "correct_answer": "invalid"
  },

  prompt_injection_text =
  "Be careful with disjunction. From B(x)->not M(x) or not E(x), E(x) does not imply not B(x),
  because not M(x) may still satisfy the premise."
)

CentralMemoryRecord(
  memory_id = "GM-I-008",
  pool_type = "task_insight_pool",
  memory_type = "reusable_high_level_insight",
  similarity_to_current_task = 0.86,

  insight_statement =
  "Deductive invalidity can be established by one countermodel.",

  insight_scope = [
    "valid/invalid classification",
    "formal logic tasks",
    "natural-language arguments",
    "universal or conditional premises"
  ],

  why_it_matters =
  "The solver does not need to derive the opposite conclusion. It only needs to show that the
  stated premises leave at least one possible world where the conclusion is false.",

  prompt_injection_text =
  "Do not over-prove. If one possible assignment makes the premises true and the conclusion false,
  classify the argument as invalid."
)

CentralMemoryRecord(
  memory_id = "GM-E-024",
  pool_type = "evaluator_feedback_pool",
  memory_type = "evaluator_preference_and_scoring_signal",
  similarity_to_current_task = 0.72,

  evaluator_feedback = {
    "high_score_pattern": "Answer defines predicates, writes the formal premises, and provides a concrete countermodel.",
    "low_score_pattern": "Answer only outputs valid/invalid without explaining whether premises entail the conclusion.",
    "preferred_final_style": "Concise final label after enough validation evidence has been accumulated."
  },

  prompt_injection_text =
  "For the validation stage, include predicate definitions and a concrete countermodel.
  For the final executor stage, compress the result to the output label."
)
\end{redlog}
\end{gmemcontent}

\gmemsection{Memory Reading Process}
\begin{gmemcontent}
\begin{redlog}
GMemoryMergedPacket(
  target_framework = "AutoGen",
  target_stage = "L1",
  target_role = "solver",

  retrieved_records = [
    "successful": ["GM-S-017"],
    "failed": ["GM-F-011"],
    "insights": ["GM-I-008"],
    "evaluator_feedback": ["GM-E-024"]
  ],

  centralized_memory_summary = {
    "task_family": "deductive validity",
    "dominant_success_strategy": "predicate formalization + countermodel search",
    "dominant_failure_mode": "over-strengthening disjunctive premises",
    "answer_format_constraint": "must output valid or invalid"
  },

  positive_strategy = [
    "Map ingredient relations to predicates.",
    "Write the first premise as E(VitaminA).",
    "Write the second premise as forall x, B(x) -> not M(x) or not E(x).",
    "Write the conclusion as not B(VitaminA).",
    "Try the assignment E(VitaminA)=true, B(VitaminA)=true, M(VitaminA)=false."
  ],

  negative_constraints = [
    "Do not use external knowledge about actual cosmetic ingredients.",
    "Do not infer not B(VitaminA) merely from E(VitaminA).",
    "Do not collapse a disjunction into a conjunction or a stronger exclusion."
  ],

  validation_checklist = [
    "Premise 1 is true under the candidate model.",
    "Premise 2 is true because not M(VitaminA) is true.",
    "The conclusion not B(VitaminA) is false.",
    "Therefore, a countermodel exists and the argument is invalid."
  ],

  prompt_injected_to_solver =
  "Central memory suggests that this task is a deductive-validity problem.
  Use only the stated premises. Formalize the predicates and search for a countermodel.
  Be especially careful with the disjunction in the second premise:
  E(VitaminA) being true does not force B(VitaminA) to be false, because not M(VitaminA)
  can still satisfy the premise."
)
\end{redlog}

\vspace{0.5em}
\textbf{Memory packet injected to Solver Agent 0:}

\vspace{0.4em}
\textbf{Positive guidance:}

Formalize the premises with predicates.

Try to construct a countermodel.

If one model makes all premises true and the conclusion false, answer invalid.

\vspace{0.5em}
\textbf{Negative guidance:}

Do not treat ``not M(x) or not E(x)'' as ``not B(x)''.

Do not use outside facts about cosmetics or ingredients.

\vspace{0.7em}
\textbf{[\#1] Solver Agent 0}

I receive the task plus the G-memory packet. I reason explicitly over the premises.

\vspace{0.4em}
\textbf{Output:}

The argument is invalid.

\end{gmemcontent}

\gmemsection{Memory Update Process}
\begin{gmemcontent}
\begin{redlog}
[G-memory Final Update]
Write the completed trajectory back into centralized memory:
task_signature = "deductive validity with universal disjunctive premise"
task_family = "BBH logical deduction"
framework = "AutoGen"
final_label = "invalid"
final_confidence = high

finalized_task_record = {
  "root_question": "Does the conclusion logically follow from the stated premises?",
  "final_answer": "invalid",
  "successful_strategy": "predicate formalization + countermodel construction",
  "dominant_failure_avoided": "collapsing the disjunction into a stronger exclusion",
  "key_countermodel": {
    "E(VitaminA)": true,
    "B(VitaminA)": true,
    "M(VitaminA)": false
  },
  "why_invalid": "All premises can be true while the conclusion is false."
}

finalized_stage_trace = [
  {
    "agent_role": "solver",
    "node_id": "GM-current-L1-agent0",
    "main_contribution": "identified that a countermodel may exist",
    "score": 9.0
  },
  {
    "agent_role": "validator",
    "node_id": "GM-current-L2-agent1",
    "main_contribution": "constructed the explicit symbolic countermodel",
    "score": 10.0
  },
  {
    "agent_role": "executor",
    "node_id": "GM-current-L3-agent2",
    "main_contribution": "compressed the validated reasoning into the final label",
    "score": 9.0
  }
]

extracted_reusable_knowledge = {
  "success_rule_1": "Formalize natural-language relations as predicates before deciding validity.",
  "success_rule_2": "For deductive invalidity, one countermodel is sufficient.",
  "success_rule_3": "When a premise contains a disjunction, test whether the alternative disjunct can still satisfy the premise.",
  "failure_rule_1": "Do not infer not B(a) from B(a)->not M(a) or not E(a) merely because E(a) is true."
}

evaluator_scores = {
  "L1": 9.0,
  "L2": 10.0,
  "L3": 9.0,
  "trajectory_average": 9.33
}

Update centralized pools and indexes:
  successful_trajectory_pool += {
    "memory_id": "GM-S-new-021",
    "source_task_signature": "deductive validity with universal disjunctive premise",
    "stored_trace": "solver -> validator -> executor trajectory",
    "stored_countermodel": "E(VitaminA)=true, B(VitaminA)=true, M(VitaminA)=false",
    "stored_label": "invalid"
  }

  failed_trajectory_pool += {
    "memory_id": "GM-F-new-021",
    "stored_failure_pattern": "incorrectly collapsing a disjunctive conditional",
    "stored_negative_constraint": "one false disjunct does not invalidate the whole premise if the other disjunct remains true",
    "link_to_corrected_case": "GM-S-new-021"
  }

  task_insight_pool += {
    "memory_id": "GM-I-new-021",
    "insight": "For premise B(x)->not M(x) or not E(x), E(a) does not imply not B(a), because not M(a) can satisfy the disjunction.",
    "generalization_scope": "natural-language deductive validity tasks with disjunctive premises"
  }

  evaluator_feedback_pool += {
    "memory_id": "GM-E-new-021",
    "preferred_high_score_pattern": "define predicates, write the formal premise, and provide a concrete countermodel",
    "observed_score_trace": "[L1: 9.0, L2: 10.0, L3: 9.0]"
  }

agent_message_graph += {
  "new_nodes": [
    "task_root_node",
    "GM-current-L1-agent0",
    "GM-current-L2-agent1",
    "GM-current-L3-agent2",
    "countermodel_memory_node",
    "final_label_node"
  ],
  "new_edges": [
    "task_root_node -> GM-current-L1-agent0",
    "GM-current-L1-agent0 -> GM-current-L2-agent1",
    "GM-current-L2-agent1 -> countermodel_memory_node",
    "countermodel_memory_node -> GM-current-L3-agent2",
    "GM-current-L3-agent2 -> final_label_node"
  ]
}

Refresh retrieval weights and summaries:
Increase weight of GM-S-017 because its countermodel strategy directly supported a successful trajectory.
Increase weight of GM-I-008 because the "one countermodel is sufficient" rule was validated again.
Preserve GM-F-011 as a high-value negative constraint because it matches the core failure mode of this task.
Regenerate centralized summary so future solver agents see this task family as predicate-first, countermodel-first, disjunction-sensitive.
\end{redlog}
\end{gmemcontent}

\end{gmemorycasebox}

\section{Prompt Set}


\definecolor{promptOrangeTitle}{HTML}{FCE4C8}
\definecolor{promptOrangeBar}{HTML}{FFF3E6}
\definecolor{promptOrangeBorder}{HTML}{D9822B}
\definecolor{promptCodeBg}{HTML}{FFF9F2}
\definecolor{promptCodeText}{HTML}{7A3E00}

\newtcolorbox{prompttablebox}[1]{
    enhanced,
    breakable,
    colback=white,
    colframe=promptOrangeBorder,
    colbacktitle=promptOrangeTitle,
    coltitle=black,
    fonttitle=\bfseries\large,
    title={#1},
    boxrule=0.8pt,
    arc=2pt,
    left=6pt,
    right=6pt,
    top=6pt,
    bottom=6pt,
    toptitle=6pt,
    bottomtitle=6pt,
    lefttitle=8pt,
    righttitle=8pt,
    before skip=10pt,
    after skip=12pt,
    borderline west={2pt}{0pt}{promptOrangeBorder}
}

\newcommand{\promptrowtitle}[2]{%
    \par\vspace{0.55em}
    \noindent
    \colorbox{promptOrangeBar}{%
        \parbox{\dimexpr\linewidth-2\fboxsep\relax}{%
            \textbf{#1}\hfill{\small #2}%
        }%
    }%
    \par\vspace{0.35em}
}

\newenvironment{promptcode}{%
  \VerbatimEnvironment
  \begingroup
  \renewcommand{\FancyVerbFormatLine}[1]{\textcolor{promptCodeText}{##1}}%
  \begin{Verbatim}[
    fontsize=\scriptsize,
    breaklines=true,
    breakanywhere=true,
    frame=single,
    framesep=4pt,
    rulecolor=\color{promptOrangeBorder},
    fillcolor=\color{promptCodeBg},
    formatcom=\color{promptCodeText}
  ]%
}{%
  \end{Verbatim}%
  \endgroup
}

\begin{prompttablebox}{\ours Prompt}

\promptrowtitle{1. Exploration Pool Prompt}

When the exploration memory pool is selected, no historical memory fragment is reused.
Instead, the agent enters a fresh exploration mode and solves the task through
standard workflow prompts, including role definition, approach decision, optional
problem decomposition, direct problem solving, and solution integration.

\promptrowtitle{1.1 Role-Definition Prompt}{Role selection}

\begin{promptcode}
Given this problem:
{problem_description[:300]}

Select the BEST role from this list to solve this problem:
- Problem Solver: General problem-solving and analysis
- Data Analyst: Analyzing tables, numerical data, patterns
- Reasoning Specialist: Logical deduction and causal reasoning
- Verification Checker: Verifying correctness and validating solutions
- Planning Strategist: Breaking down complex tasks into steps
- Math Solver: Mathematical calculations and equations
- Logic Checker: Checking logical consistency and correctness
- Table Interpreter: Reading and interpreting data tables
- Causal Analyst: Analyzing cause-and-effect relationships
- Navigation Tracker: Tracking positions and movements

Your response MUST be EXACTLY ONE of these role names:
Problem Solver, Data Analyst, Reasoning Specialist, Verification Checker,
Planning Strategist, Math Solver, Logic Checker, Table Interpreter,
Causal Analyst, Navigation Tracker

Only respond with the role name, nothing else.
\end{promptcode}

\promptrowtitle{1.2 Approach-Decision Prompt}{Routing decision}

\begin{promptcode}
Analyze this problem and decide the BEST approach:

Problem: {problem_desc[:200]}...

{if decomposition_score > 0.4:
  "HINT: This problem shows signs of complexity that would benefit from decomposition"}

Decision criteria for 'problem_decomposition':
- Can be broken into 2-4 clear, independent sub-tasks
- Requires different types of expertise or analysis
- Has multiple distinct components or phases
- Examples: system design, multi-step processes, building components

Decision criteria for 'direct_answer':
- Can be solved straightforwardly with a single approach
- Doesn't require breaking into sub-parts
- Simple analysis, calculation, or explanation



Choose the approach that will lead to the BEST solution quality.
Respond ONLY with ONE word: problem_decomposition, direct_answer.
\end{promptcode}

\promptrowtitle{1.3 Direct-Solving Prompt}{Direct answer}

\begin{promptcode}
Solve this problem:
{problem_description}

Provide a clear, direct answer.
\end{promptcode}

\promptrowtitle{1.4 Problem-Decomposition Prompt}{Sub-task generation}

\begin{promptcode}
You are an expert at breaking down complex problems into manageable sub-tasks.

Original Problem:
{problem_desc}

Please decompose this problem into 2-4 DISTINCT and INTERDEPENDENT sub-problems.

Requirements:
1. Each sub-problem should be focused on a specific aspect or step.
2. Sub-problems should be solvable with different expertise levels.
3. Each must contribute to solving the original problem.
4. Ensure the sub-problems are complementary and cover different angles.

For each sub-problem, provide:
- "id": Sequential ID, such as 1, 2, 3.
- "description": Clear, specific description of the sub-problem.
- "focus": Main focus area, e.g., "Analysis", "Design", "Verification".
- "dependencies": Dependencies on other sub-problems, or an empty list.

Respond ONLY with a valid JSON array like this:
[
  {"id": 1, "description": "...", "focus": "...", "dependencies": []},
  {"id": 2, "description": "...", "focus": "...", "dependencies": [1]}
]
\end{promptcode}

\promptrowtitle{2. Exploitation-Pool Prompt with Similarity Matching}{Historical memory reuse}

When the Exploitation-Pool is selected, the framework does not immediately inject
historical memory. Instead, it first applies a similarity-matching mechanism.
The current task description is used as the retrieval query, encoded into an
embedding, and compared against stored memory fragments. Only
fragments whose semantic similarity exceeds a predefined threshold are reused.

\promptrowtitle{2.1 Exploitation-Pool Retrieval Query}{Similarity search}

\begin{promptcode}
Use the current task description as the retrieval key and search the  memory pool for semantically similar historical fragments.

Retrieval query:
{problem_description}

Search scope:
agent_id = {agent_id}
stage_id = {stage_id}
pool_type = "Exploitation-Pool"

Return candidate memory fragments ranked by semantic similarity.
\end{promptcode}

\promptrowtitle{2.2 Similarity Filtering Prompt}{Threshold selection}

\begin{promptcode}
Given the current task and retrieved memory candidates, decide which historical
memory fragments are relevant enough to reuse.

Current task:
{problem_description}

Candidate memories:
{retrieved_memory_candidates}

Similarity threshold:
{similarity_threshold}

Selection criteria:
- Keep fragments whose semantic similarity exceeds the threshold.
- Prefer fragments that match the task type, reasoning pattern, and expected output format.
- Discard fragments that only share surface-level keywords.
- Discard fragments that may introduce misleading assumptions.

Return ONLY the selected memory fragment IDs as a JSON array:
["memory_id_1", "memory_id_2"]
\end{promptcode}

\promptrowtitle{2.3 Exploitation-Pool Prompt}{Memory-conditioned solving}

\begin{promptcode}
You are solving a new task with help from relevant historical memory fragments.

Current task:
{problem_description}

Retrieved Exploitation-Pool fragments:
{selected_Exploitation-Pool_fragments}

Use the retrieved memory only as guidance. Do not copy previous answers directly.
Adapt useful reasoning patterns, checks, or constraints to the current task.

Please solve the current task and provide a clear final answer.
\end{promptcode}

\promptrowtitle{3. Evaluation Prompt}{Stage-level quality scoring}

The evaluation framework scores the quality of execution at the stage level.
It considers both the integrated solutions and the raw direct LLM answers, and
returns structured feedback for subsequent memory updates.

\promptrowtitle{3.1 Evaluation Prompt}{Evaluator instruction}

\begin{promptcode}
You are an expert evaluator. Evaluate the overall quality of work done in this stage.

PROBLEM:
{original_problem}

STAGE: {stage} - {stage_name}
- Number of tasks: {num_nodes}
- Task node paths: {node_paths}
- Agents involved: {agents_involved}
- Action types distribution: {action_types_count}

SOLUTIONS PROVIDED (Final integrated solutions):
{solutions_summary}

DIRECT LLM ANSWERS (Raw LLM responses before processing):
{direct_llm_answers_summary}
\end{promptcode}

\promptrowtitle{3.2 Stage-Specific Criteria for $t_1$}{Initial stage}

\begin{promptcode}
1. Problem Understanding:
Did the agent properly understand the problem?

2. Decomposition Quality:
If decomposed, is the breakdown logical and complete?

3. Solution Clarity:
Are the solutions clear and well-structured?

4. LLM Direct Answer Quality:
Is the LLM's direct response accurate and helpful?

5. Foundation:
Did this stage provide good foundation for next stages?
\end{promptcode}

\promptrowtitle{3.3 Stage-Specific Criteria for $t_2$}{Intermediate stage}

\begin{promptcode}
1. Processing Quality:
How well were intermediate tasks solved?

2. Building on Previous:
Did agents effectively use guidance from stage t_1?

3. Task Allocation:
Were tasks appropriately allocated to capable agents?

4. Coherence:
Do the solutions form a coherent middle layer?

5. LLM Answer Consistency:
Do the LLM direct answers align with the integrated solutions?
\end{promptcode}

\promptrowtitle{3.4 Stage-Specific Criteria for $t_3$}{Final stage}

\begin{promptcode}
1. Refinement Quality:
How well were solutions refined?

2. Integration:
How well do the final solutions integrate all previous work?

3. Completeness:
Is the final solution complete and comprehensive?

4. Excellence:
Does the final work meet high quality standards?

5. LLM Answer Quality:
Are the LLM direct answers comprehensive and accurate?
\end{promptcode}

\promptrowtitle{3.5 Expected Evaluator Output}{Structured feedback}

\begin{promptcode}
{
  "score": <0-10>,
  "stage_quality": "<poor/fair/good/excellent>",
  "reasoning": "<detailed explanation>",
  "solution_quality": "<assessment of the integrated solutions>",
  "llm_answer_quality": "<assessment of the LLM direct answers>",
  "strengths": "<what went well>",
  "weaknesses": "<what could be improved>",
  "agent_coordination": "<how well agents worked together>"
}
\end{promptcode}

\end{prompttablebox}


\newpage

\end{document}